\documentclass{article}
\usepackage{graphicx}
\usepackage{epsfig}
\usepackage{amssymb}
\usepackage{amsthm}
\usepackage{amsfonts}
\usepackage{amsmath}
\usepackage{multicol,multirow}
\usepackage{algorithm,algcompatible}
\usepackage[margin=0.95in]{geometry}
\usepackage{enumitem}
\usepackage{float}
\usepackage{booktabs}
\newtheorem{theorem}{Theorem}
\newtheorem{lemma}[theorem]{Lemma}
\newtheorem{example}[theorem]{Example}

\newtheorem{proposition}[theorem]{Proposition}
\newtheorem{defn}{Definition}

\newtheorem{assumption}{Assumption}
\newtheorem{remark}{Remark}

\newcommand{\E}{\mathbb{E}}
\newcommand{\FD}{\mathrm{FD}}
\newcommand{\TD}{\mathrm{TD}}
\usepackage{amsthm}
\usepackage{tikz}
\usepackage{lipsum}

\newcommand{\setS}{\mathcal{S}}
\newcommand{\cc}{\mathcal{C}}
\newcommand{\R}{\mathbb{R}}
\newcommand{\id}{\mathrm{Id}}
\newcommand{\Lp}{\mathcal{L}}
\newcommand{\rk}{\mathrm{rank}}

\newcommand{\hatx}{\hat{x}}
\newcommand{\Proj}{\mathcal{P}}
\newcommand{\cg}{\mathcal{G}}
\usepackage{subfiles}
\usepackage{multirow}
\usepackage{makecell}
\usepackage{parskip}
\usepackage[font=small]{subfig}

\newcommand{\tr}{\mathrm{trace}}

\newcommand{\xstar}{x^\star}
\def\thm@space@setup{%
  \thm@preskip=0cm plus 1cm minus 2cm
  \thm@postskip=\thm@preskip 
}

\usepackage[
            CJKbookmarks=true,
            bookmarksnumbered=true,
            bookmarksopen=true,
            colorlinks=true,
            citecolor=blue,
            linkcolor=black,
            anchorcolor=red,
            urlcolor=blue,
            ]{hyperref}
\usepackage{cleveref}

\DeclareMathOperator*{\argmax}{argmax} 
\DeclareMathOperator*{\argmin}{argmin} 

\usepackage{authblk}

\title{Model Selection over Partially Ordered Sets} 
\author[$a$]{Armeen Taeb}
\author[$b$]{Peter B\"uhlmann}
\author[$c$]{Venkat Chandrasekaran}
\affil[$a$]{Department of Statistics, University of Washington}
\affil[$b$]{Seminar for Statistics, ETH Z\"urich}
\affil[c]{Departments of Computing and Mathematical Sciences and of Electrical Engineering, California Institute of Technology}

\date{}
\begin{document}
\maketitle
\begin{abstract}
In problems such as variable selection and graph estimation, models are characterized by Boolean logical structure such as presence or absence of a variable or an edge.  Consequently, false positive error or false negative error can be specified as the number of variables/edges that are incorrectly included or excluded in an estimated model.  However, there are several other problems such as ranking, clustering, and causal inference in which the associated model classes do not admit transparent notions of false positive and false negative errors due to the lack of an underlying Boolean logical structure.  In this paper, we present a generic approach to endow a collection of models with partial order structure, which leads to a hierarchical organization of model classes as well as natural analogs of false positive and false negative errors.  We describe model selection procedures that provide false positive error control in our general setting and we illustrate their utility with numerical experiments.
\medbreak
\end{abstract}
\section{Introduction}
{I}n data-driven approaches to scientific discovery, one is commonly faced with the problem of model selection.  Popular examples include variable selection 
(which covariates predict a response?) and graph estimation (which pairs of variables have nonzero correlation or partial correlation?).  As exemplified by these two problems, a common feature of most model selection problems in the literature is that the collection of models is organized according to some type of Boolean logical structure, such as the presence versus absence of a variable or an edge.  A consequence of such structure is that model complexity can be conveniently specified as the number of attributes (variables or edges) in a model, while false positive error or false negative error corresponds to the number of attributes that are incorrectly included or excluded in the model.

In many contemporary applications, models represent a far richer range of phenomena that are not conveniently characterized via Boolean logical structure. As a first example, suppose we are given observations of covariate-response pairs and we wish to order the covariates based on how well they predict a response; the collection of models is given by the set of rankings of the covariates.  Second, consider a clustering problem in which we are given observations of a collection of variables and the goal is to group them according to some measure of affinity, with the number of groups and the number of variables assigned to each group not known a priori; here the model class is given by the collection of all possible partitions of the set of variables. Third, suppose we wish to identify causal relations underlying a collection of variables; the model class is the set of completed partially directed acyclic graphs. Finally, consider the blind source separation problem in which we are given a signal expressed as an additive combination of source signals and our objective is to identify the constituent sources, without prior information about the number of sources or their content; here the model class is the collection of all possible linearly independent subsets of vectors. 

In these preceding examples, we lack a systematic definition of model complexity, false positive error, and false negative error due to the absence of Boolean logical structure in each collection of models. In particular, in the first three examples, valid models are characterized by structural properties such as transitivity, set partitioning, and graph acyclicity, respectively; these properties are global in nature and are not concisely modeled via separable and local characteristics such as an attribute (a variable or edge) being included in a model independently of other attributes.  In the fourth example of blind source separation, false positive and false negative errors should not be defined merely via the inclusion or exclusion of true source vectors in an estimated set but should instead consider the degree of alignment between the estimated and true sources, which again speaks to the lack of a natural Boolean logical structure underlying the associated model class.

As a concrete illustration of the inappropriateness of Boolean logical structure for clustering, consider three items $a, b, c$, with the null model given by the three clusters $\{a\}, \{b\}, \{c\}$, the true model by the two clusters $\{a,b\}, \{c\}$, and the estimated model by the single cluster $\{a,b,c\}$.  An incorrect perspective grounded in Boolean logical structure suggests a false positive error of two, with the `false discoveries' being that $c$ is in the same cluster as $a$ and as $b$.  However, accounting for set partition structure yields the more accurate false positive error value of one as $a$ and $b$ are in the same cluster in the true and estimated models; hence, including $c$ in the same cluster as $\{a,b\}$ should only incur one false discovery.  

While the preceding four problems have been studied extensively, the associated methods do not systematically control false positive error as this quantity is not formally defined.  Selection procedures that yield models with small false positive error play an important role in data-driven methods for gathering evidence, rooted in the empirical philosophy and statistical testing foundations of falsification of theories and hypotheses \cite{Neyman1928ONTU,Fisher,Poper}.

\subsection{Our Contributions}
We begin in Section~\ref{sec:framework} by describing how collections of models may be endowed with the structure of a partially ordered set (poset).  Posets are relations that satisfy reflexivity, transitivity, and antisymmetry, and they facilitate a hierarchical organization of a set of models that leads to a natural definition of model complexity.  Building on this framework, we develop an axiomatic approach to defining functions over poset element pairs for evaluating similarity.  This yields generalizations of well-known measures such as family-wise error and false discovery rate to an array of model selection problems in the context of ranking, causal inference, multiple change-point estimation, clustering, multi-sample testing, and blind source separation.  In Section~\ref{sec:false_discovery_control}, we describe two generic model selection procedures that search over poset elements in a greedy fashion and that provide false discovery control in discrete model posets. 
The first method is based on subsampling and model averaging and it builds on the idea of stability selection \cite{Meins2010Stability,Shah2013Stability} for the variable selection problem, while the second method considers a sequence of hypothesis tests between models of growing complexity.  With both these methods, the combinatorial properties of a model poset play a prominent role in determining computational and statistical efficiency.  Proofs of the theorems of Section~\ref{sec:false_discovery_control} are provided in Section~\ref{sec:proofs}.  In Section~\ref{sec:experiments} we provide numerical illustration via experiments on synthetic and real data. {The code for implementing our methods is available at \url{https://github.com/armeentaeb/model-selection-over-posets}}.


\subsection{Related Work}
\label{sec:related}
Classic approaches to model selection such as the AIC and BIC assess and penalize model complexity by counting the number of attributes included in a model \cite{AIC,BIC}.  More generally, such complexity measures facilitate a hierarchical organization of model classes, and this perspective is prevalent throughout much of the model selection literature \cite{Goeman2008MultipleTO,Meinshausen2008HierarchicalTO,Yekutieli2008HierarchicalFD,Rosenbaum2008TestingHI,Foster2008investingAP,BH}. However, these complexity measures rely on a Boolean logical structure underlying a collection of models, and are therefore not well-suited to model classes that are not characterized in this manner.  The poset formalism presented in this paper is sufficiently flexible to facilitate model selection over model classes that are more complex than those characterized by Boolean logical structure (such as the illustration presented previously with clustering, see also Example~\ref{ex:clustering}), while being sufficiently structured to permit precise definitions of model complexity as well as false positive and false negative errors.

\section{Poset Framework for Model Selection}
\label{sec:framework}
We begin by describing how collections of models arising in various applications may be organized as posets.  Next, we present approaches to endow poset-structured models with suitable notions of true and false discoveries.

\subsection{Model Classes as Posets}
\label{sec:order_theoretic_formulation}

We begin with some basics of posets.  A \emph{poset} $(\Lp, \preceq)$ is a collection $\Lp$ {of elements} and a relation $\preceq$ that is reflexive ($x \preceq x, ~ \forall x \in \Lp$), transitive ($x \preceq y, y \preceq z \Rightarrow x \preceq z, ~ \forall x,y,z \in \Lp$), and anti-symmetric ($x \preceq y, y \preceq x \Rightarrow x = y, ~ \forall x,y \in \Lp$).  An element $y \in \Lp$ \emph{covers} $x \in \Lp$ if $x \preceq y$, $x \neq y$, and there is no $z \in \Lp \backslash \{x,y\}$ with $x \preceq z \preceq y$; we call such $(x,y)$ a \emph{covering pair}.  A \emph{path} from $x_1 \in \Lp$ to $x_k \in \Lp$ is a sequence $(x_1, \dots, x_k)$ with $x_2,\dots,x_{k-1} \in \Lp$ such that $x_i$ covers $x_{i-1}$ for each $i=2,\dots,k$.  Throughout this paper, we focus on posets in which there is a \emph{least element}, i.e., an element $x_{\text{least}} \in \Lp$ such that $x_{\text{least}} \preceq y$ for all $y \in \Lp$; such least elements are necessarily unique.  Finally, a poset is \emph{graded} if there exists a function $\rk(\cdot)$ mapping poset elements to the nonnegative integers such that the rank of the least element is $0$ and $\rk(y) = \rk(x) + 1$ for $y \in \Lp$ that covers $x \in \Lp$.  In graded posets with least elements, each path from the least element to any $x \in \Lp$ has length equal to $\rk(x)$.  Posets are depicted visually using Hasse diagrams in which a directed arrow is drawn from $x \in \Lp$ to any $y \in \Lp$ that covers $x$.

Posets offer an excellent framework to formulate model selection problems as model classes in many applications possess rich partial order structures.  In particular, the poset-theoretic quantities in the preceding paragraph have natural counterparts in the context of model selection -- the least element corresponds to the `null' model that represents no discoveries, the relation $\preceq$ specifies a notion of containment between simpler and more complex models, and the rank function serves as a measure of model complexity that respects the underlying containment relation.  We present several concrete illustrations next; Figure \ref{fig:1} presents Hasse diagrams associated with several of these examples.  

\captionsetup[figure]{font=footnotesize}
\begin{figure*}[h!]
    \centering
    \subfloat[{\footnotesize{Variable selection}}]{\includegraphics[width=.17\columnwidth]{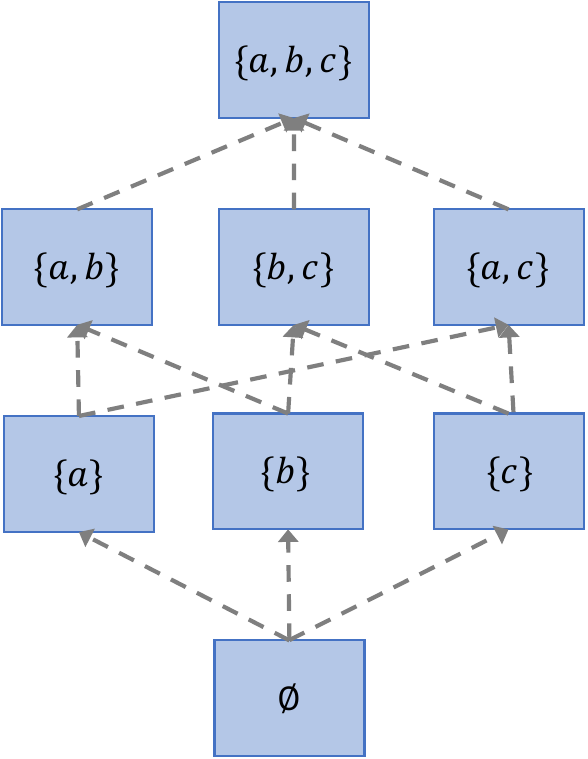}}~
    \subfloat[{\footnotesize{Clustering}}]{\includegraphics[width=.4\columnwidth]{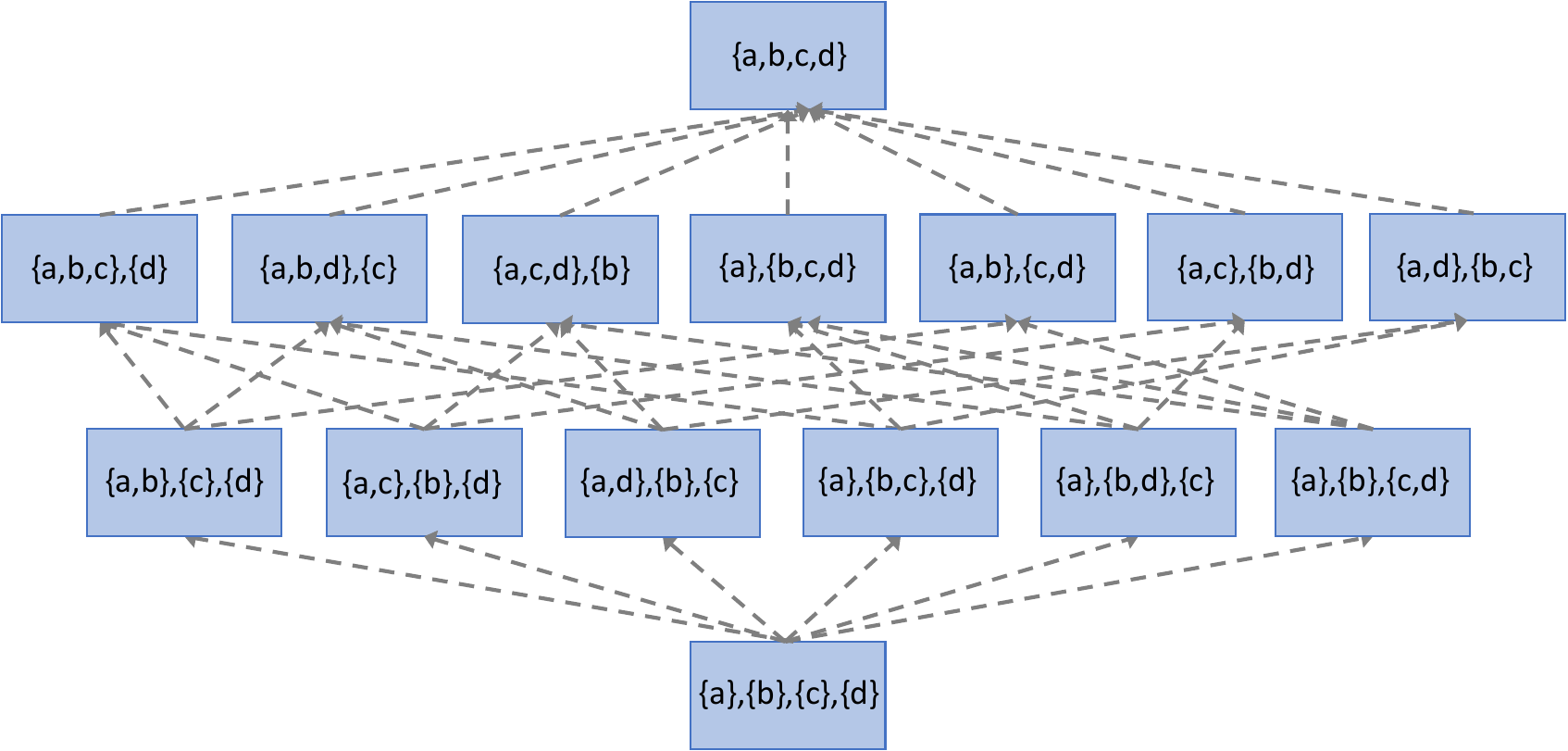}}~
    \subfloat[{\footnotesize{Multisample testing}}]{\includegraphics[width=.4\columnwidth]{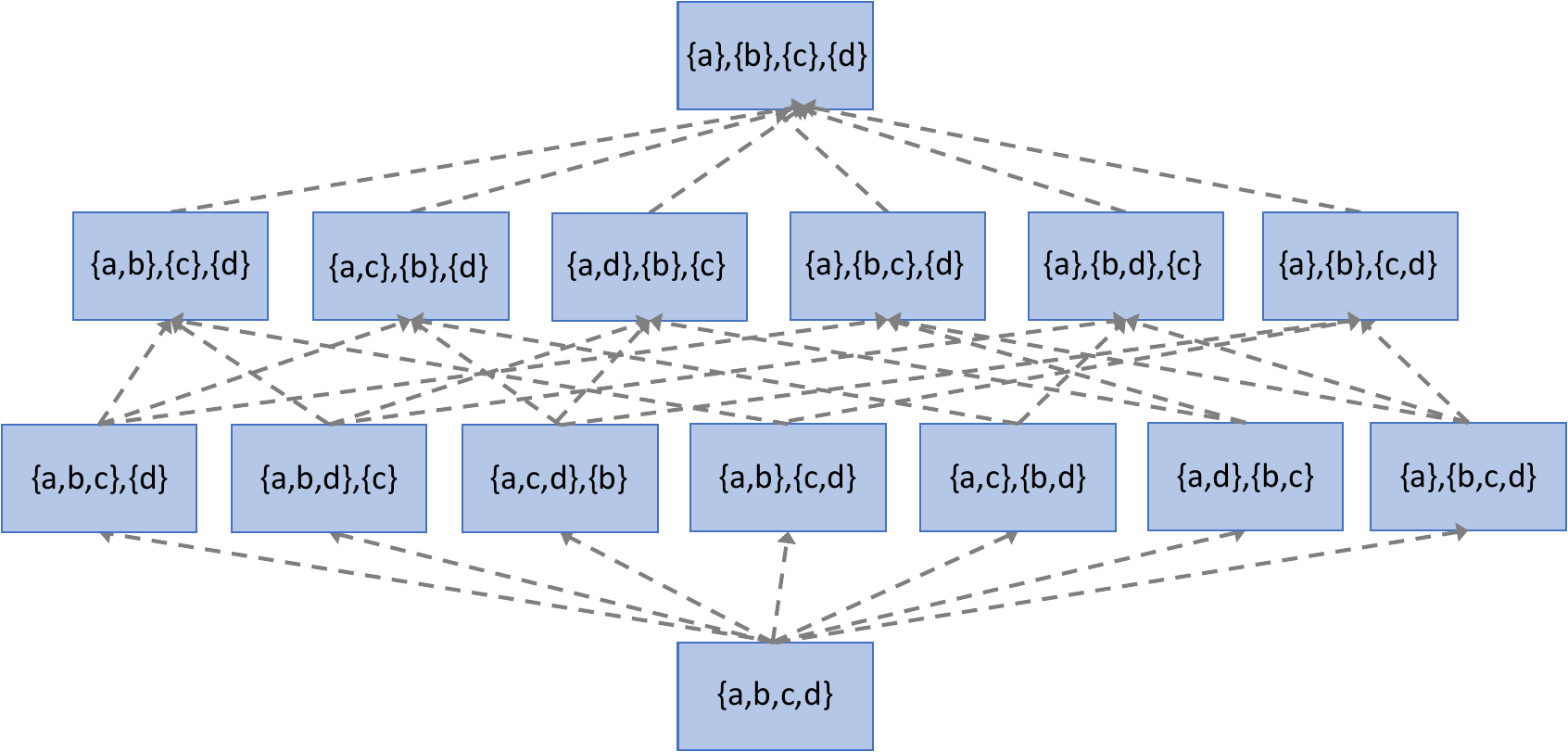}}
\quad
\subfloat[\footnotesize{Causal inference}]{\includegraphics[width=.32\columnwidth]{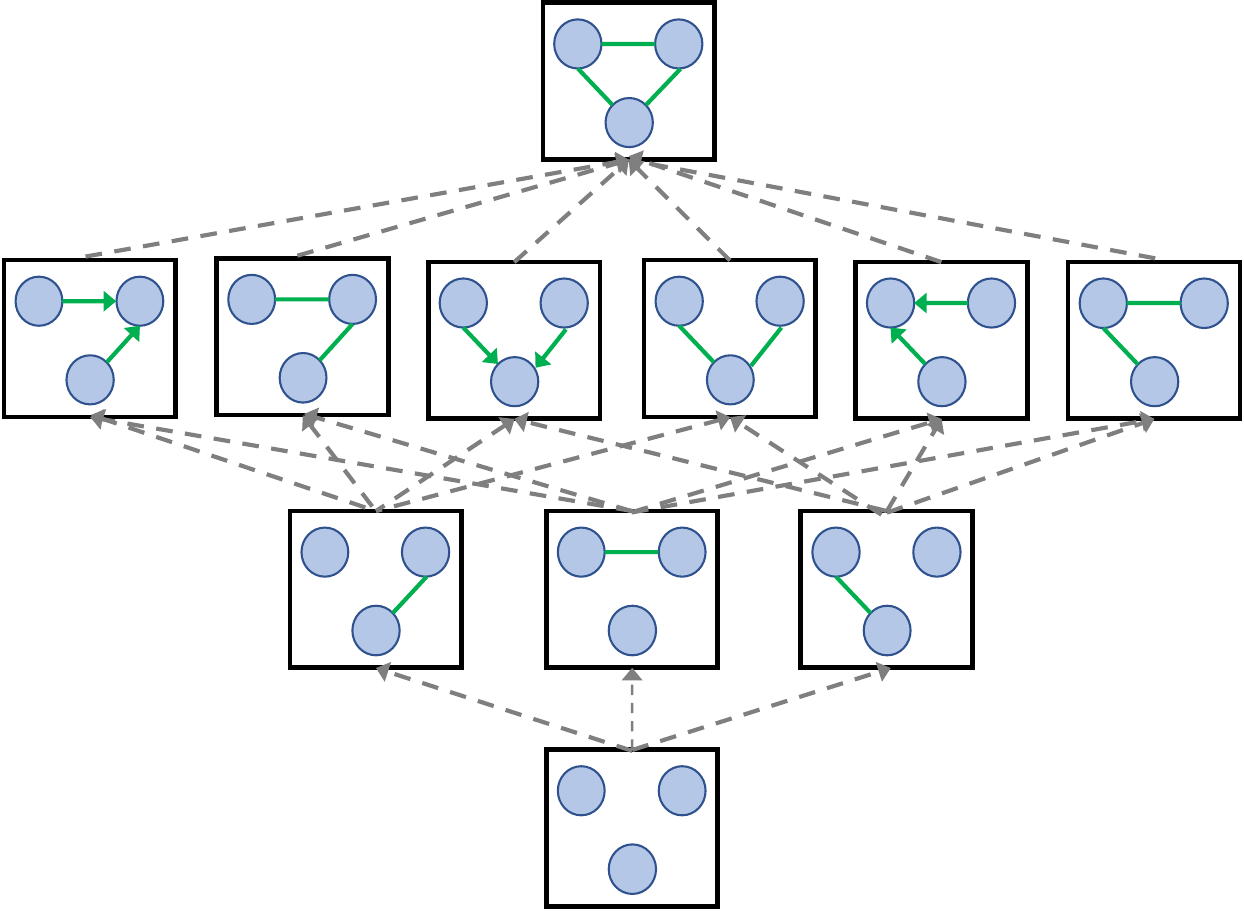}}
~~
\subfloat[\footnotesize{Partial ranking}]{\includegraphics[width=.32\columnwidth]{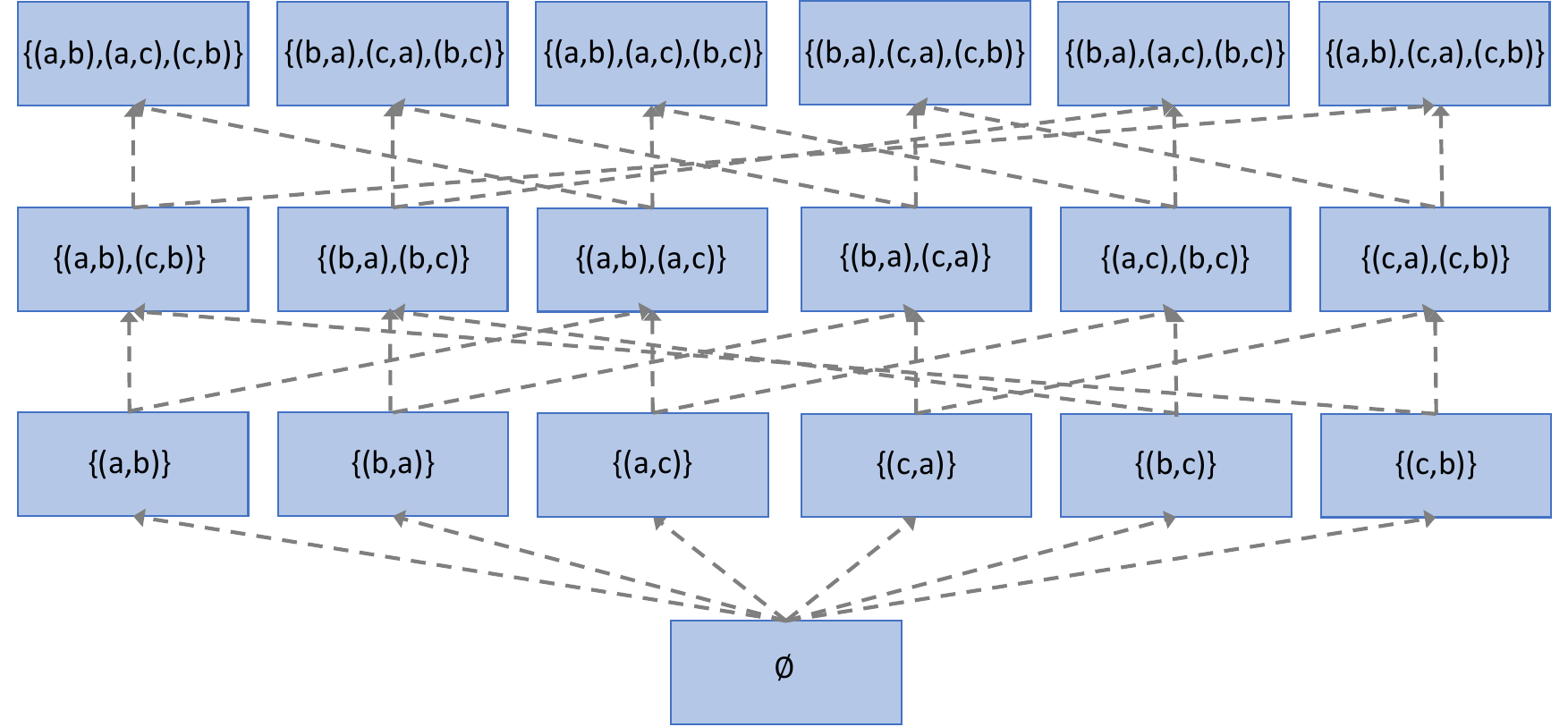}}
~~
\subfloat[\footnotesize{Total ranking}]{\includegraphics[width=.32\columnwidth]{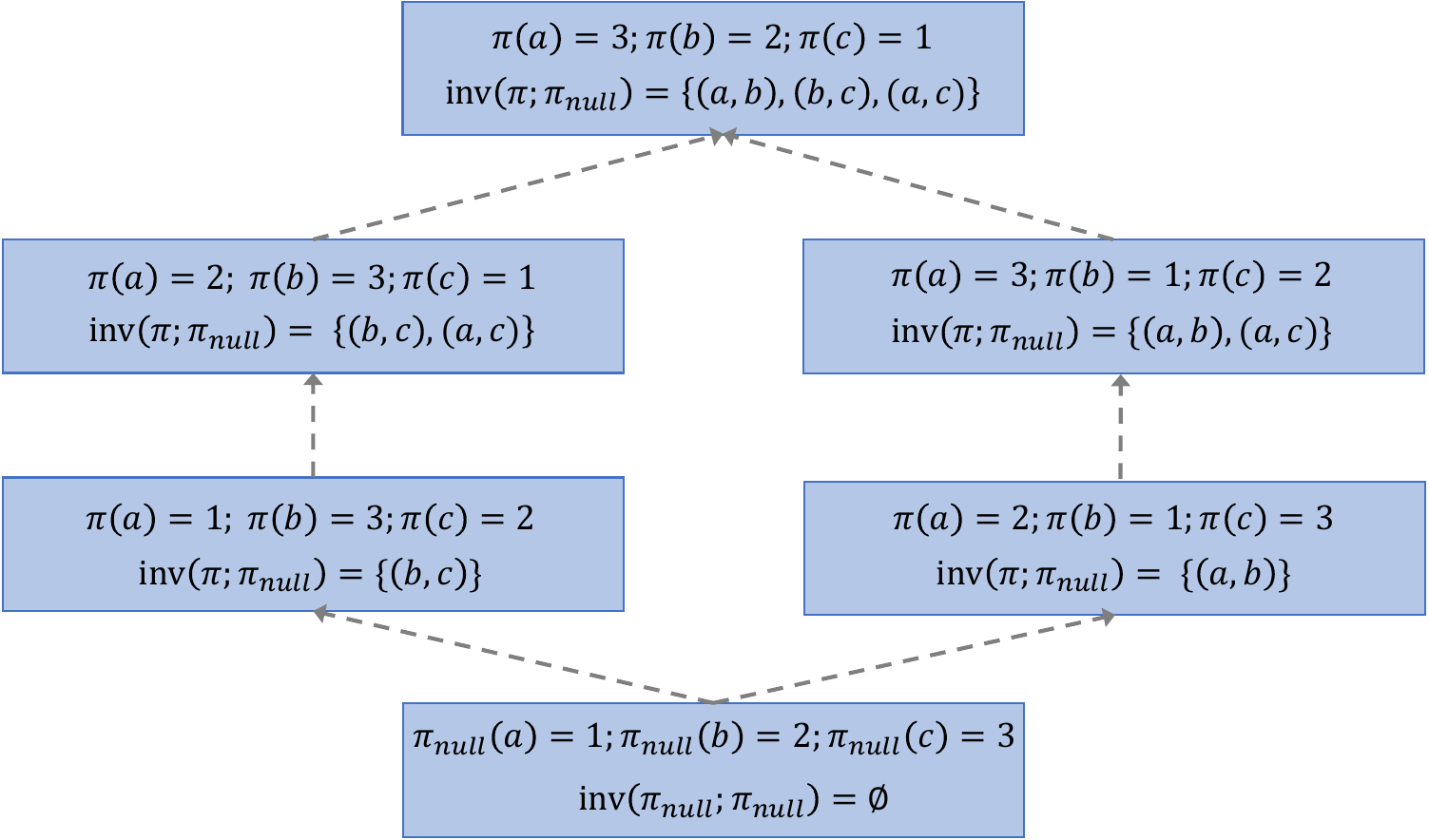}}
\caption{\small{Hasse diagrams for a) variable selection with $3$ variables (Example~\ref{ex:variable-selection}); b) clustering $4$ variables (Example~\ref{ex:clustering}); c) multisample testing with $4$ samples (Example~\ref{ex:multisample-testing}); d) causal inference with $3$ variables (Example~\ref{ex:causal-learning}); e) partial ranking of $3$ items (Example~\ref{ex:partial-ranking}); and  f) total ranking of $3$ items (Example~\ref{ex:complete-ranking})}.}
    \label{fig:1}
\end{figure*}

\begin{example}[Variable selection]\label{ex:variable-selection}
As a warm-up, consider the variable selection problem of selecting which of $p$ variables influence a response.  The poset here is the collection of all subsets of $\{1,\dots,p\}$ ordered by set inclusion, the least element is given by the empty set, and the rank of a subset is its cardinality. This poset is called the \emph{Boolean poset} \cite{stanley_2011}.
\end{example}

\begin{example}[Clustering]\label{ex:clustering}
Suppose we wish to group a collection of $p$ variables based on a given notion of similarity.  The poset here is the collection of all partitions of $\{1,\dots,p\}$ ordered by refinement, the least element is given by $p$ groups each consisting of one variable, and the rank of a partition is equal to $p$ minus the number of groups.  Thus, higher-rank elements correspond to models specified by a small number of clusters. This poset is called the \emph{partition poset} \cite{stanley_2011}.
\end{example}

\begin{example}[Multisample testing]\label{ex:multisample-testing}
As a generalization of the classic two-sample testing problem, consider the task of grouping $p$ samples with the objective that samples in a group come from the same distribution.  Although this problem is closely related to the preceding clustering problem, it is more natural for the underlying poset here to be the reverse of the partition poset that is formed by reversing the order relation of the partition poset, i.e., the poset is the collection of all partitions of $\{1,\dots,p\}$ ordered by coarsening.  With this reverse ordering, the least element corresponds to all $p$ samples belonging to the same group (i.e., coming from the same distribution), which generalizes the usual null hypothesis in two-sample testing.  The rank of a partition is equal to the number of groups minus one.  Thus, higher-rank elements correspond to the $p$ samples arising from many distinct distributions.
\end{example}

\begin{example}[Causal structure learning]\label{ex:causal-learning}
Causal associations among a collection of variables are often characterized by a directed acyclic graph (DAG), namely a graph with directed edges and no (directed) cycles, in which the nodes index the variables.  Causal structure learning entails inferring this DAG from observations of the variables.  The structure of a DAG specifies a causal model via conditional independence relations among the variables, with denser DAGs encoding fewer conditional independencies in comparison with sparser DAGs.  (See \cite{Drton2016StructureLI} for details on how the structure of a DAG encodes conditional independence relations; here we describe only those aspects that pertain to a poset formulation to organize the collection of all causal models based on graph structure.)  Distinct DAGs can specify the same set of conditional independence relations, and these are called Markov equivalent DAGs.  We introduce some terminology to characterize Markov equivalent DAGs.  The \emph{skeleton} of a DAG is the undirected graph obtained by making all the edges undirected.  A \emph{$v$-structure} is a set of three nodes $x,y,z$ such that there are directed edges from $x$ to $z$ and from $y$ to $z$, and there is no edge between $x$ and $y$.  Two DAGs are Markov equivalent if and only if they have the same skeleton and the same collection of $v$-structures \cite{verma1990equivalence}.  A Markov equivalence class of DAGs can be described by a completed partially directed acyclic graph (CPDAG), which is a graph consisting of both directed and undirected edges.  A CPDAG has a directed edge from a node $x$ to a node $y$ if and only if this directed edge is present in every DAG in the associated Markov equivalence class.  A CPDAG has an undirected edge between nodes $x$ and $y$ if the corresponding Markov equivalence class contains a DAG with a directed edge from $x$ to $y$ and a DAG with a directed edge from $y$ to $x$.  One can check that the total number of edges in a CPDAG (directed plus undirected) is equal to the number of edges in any DAG in the associated Markov equivalence class.  {The collection of CPDAGs on $p$ variables may be viewed as a poset ordered by inclusion of conditional dependencies -- CPDAGs $\mathcal{C}^{(1)}, \mathcal{C}^{(2)}$ satisfy $\mathcal{C}^{(1)} \preceq \mathcal{C}^{(2)}$ if and only if all the conditional independencies encoded by $\mathcal{C}^{(2)}$ are also encoded by $\mathcal{C}^{(1)}$, or equivalently that all the conditional dependencies encoded by $\mathcal{C}^{(1)}$ are also encoded by $\mathcal{C}^{(2)}$}.  The least element is given by the CPDAG with no edges, and the rank function is equal to the number of edges. (See Appendix~\ref{sec:graded_cpdag} for a proof that this poset is graded).  Higher-rank elements in this poset correspond to causal models exhibiting more conditional dependence relations. {In some causal inference contexts, it may be more natural to view the fully connected CPDAG as the null model. Our framework can accommodate this perspective by reversing the preceding model poset.  Specifically, in this reversed poset, the partial order is given by inclusion of conditional independencies, the least element is the fully connected CPDAG, and the rank function is $p(p-1)/2$ minus the number of edges. Higher-rank elements in this poset correspond to causal models exhibiting more conditional independence relations.}

\end{example}

\begin{example}[Multiple changepoint estimation]\label{ex:multiple-changepoint}
Consider the problem of detecting changepoints in a multivariate time series.  Specifically, we observe $p$ signals each for time instances $t=0,\dots,T-1$, each signal consists of at most one change (e.g., a change in the distribution or dynamics underlying the signal observations), and the objective is to identify these changes.  We denote changepoints via vectors $x = (x_1,\dots,x_p) \in \{0,\dots,T\}^p$, with $x_i$ denoting the time index when a change occurs in the $i$'th signal and $x_i = T$ corresponding to no change occurring.  The poset here is the set $\{0,1,\dots,T\}^p$ ordered such that $x \preceq y$ if and only if $x_j \geq y_j$ for all $j=1,\dots, p$, the least element is $(T,\dots,T)$, and the rank of an element is $p \cdot T$ minus the sum of the coordinates.  Higher-rank elements correspond to changepoint estimates in which the changes occur early.  This poset is the reverse of the (bounded) \emph{integer poset} \cite{stanley_2011} with the product order.
\end{example}

\begin{example}[Partial ranking]\label{ex:partial-ranking}
We seek a ranking of a finite set of items given noisy observations (e.g., pairwise comparisons), and we allow some pairs of items to be declared as incomparable.  Such a \emph{partial ranking} of the elements of a finite set $S$ corresponds to a \emph{strict partial order} on $S$, i.e., a relation $\mathcal{R}$ that is irreflexive $((a,a) \notin \mathcal{R}, ~ \forall a \in S)$, asymmetric $((a,b) \in \mathcal{R} \Rightarrow (b,a) \notin \mathcal{R}, ~ \forall a,b \in S)$, and transitive; if an element of $S$ does not appear in $\mathcal{R}$, then that element is incomparable to any of the other elements of $S$ in the associated partial ranking.  The poset here is the collection of strict partial orders on $S$ ordered by inclusion, the least element is the empty set, and the rank of a partial ranking is the cardinality of the associated relation.  Thus, higher-rank elements correspond to partial rankings that compare many of the covariates.
\end{example}

\begin{example}[Total ranking]\label{ex:complete-ranking}
We again wish to rank a finite collection of items but now we seek a total ranking that provides an ordered list of all the items.  The setting is that we are given a total ranking that represents our current state of knowledge (i.e., a `null model') as well as a new set of noisy observations, and the goal is to identify a total ranking that represents an update of the null model to reflect the new information.  Each \emph{total ranking} of the elements of a finite set $S$ corresponds to a one-to-one function from $S$ to the integers $\{1,\dots,|S|\}$.  Let $\pi_{\mathrm{null}}$ be the function that describes the null ranking.  A convenient way to compare total rankings and to define a poset structure over them is via the notion of an \emph{inversion set}.  For any total ranking specified by a function $\pi$, the associated inversion set (with respect to the null ranking $\pi_{\mathrm{null}}$) is defined as $\mathrm{inv}(\pi; \pi_{\mathrm{null}}) := \{(x,y) \in S \times S ~|~ \pi_{\mathrm{null}}(x) < \pi_{\mathrm{null}}(y), ~ \pi(x) > \pi(y)\}$.  The poset here (with respect to a given null ranking $\pi_{\mathrm{null}}$) is the collection of total rankings on $S$ ordered by inclusion of the associated inversion sets, the least element is the null ranking $\pi_{\mathrm{null}}$, and the rank of a total ranking is the cardinality of the associated inversion set; this rank function is also equal to the Kendall tau distance between a total ranking and $\pi_{\mathrm{null}}$.  Thus, higher-rank elements are given by total rankings that depart significantly from the null ranking $\pi_{\mathrm{null}}$.  This poset is called the \emph{permutation poset} \cite{stanley_2011}.

\end{example}

\begin{example}[Subspace estimation]\label{ex:subspace-estimation}
The task is to estimate a subspace in $\R^p$ given noisy observations of points in the subspace.  The poset is the collection of subspaces in $\R^p$ ordered by inclusion, the least element is the subspace $\{0\}$, and the rank of a subspace is its dimension.  This poset is called the \emph{subspace poset}.
\end{example}

\begin{example}[Blind source separation]\label{ex:blind-source-separation}
We are given a signal in $\R^p$ that is expressed as a linear combination of some unknown source signals and the goal is to estimate these sources.  The poset here is the collection of linearly independent subsets of unit-norm vectors in $\R^p$ ordered by inclusion, the least element is the empty set, and the rank of a linearly independent subset is equal to the cardinality of the subset.
\end{example}

{With respect to formalizing the notion of false positive and false negative errors}, Example~\ref{ex:variable-selection} is prominently considered in the literature, while Examples~\ref{ex:multisample-testing} and \ref{ex:multiple-changepoint} are multivariate generalizations of previously studied cases \cite{Hotelling1931TheGO,Lorden1971PROCEDURESFR}.  Finally, Example~\ref{ex:subspace-estimation} was studied in \cite{Taeb2020FD}, although that treatment proceeded from a geometric perspective rather than the order-theoretic approach presented in this paper.  With the exception of Example~\ref{ex:variable-selection}, none of the other examples permit a natural formulation within the traditional multiple testing paradigm due to the lack of a Boolean logical structure underlying the associated model classes.  Moreover, Examples~\ref{ex:subspace-estimation}-\ref{ex:blind-source-separation} are model classes consisting of infinitely many elements.   Nonetheless, we describe in the sequel how the poset formalism enables a systematic and unified framework for formulating model selection in all of the examples above.

\subsection{Evaluating True and False Discoveries}
\label{section:true_false_discoveries}
To assess the extent to which an estimated model signifies discoveries about the true model, we describe next a general approach to quantify the similarity between poset elements in a manner that respects partial order structure.

\begin{defn}[similarity valuation] Let $(\Lp,\preceq,\rk(\cdot))$ be a graded poset.  A function $\rho: \Lp\times\Lp \to \mathbb{R}$ that is symmetric, i.e., $\rho(x,y) = \rho(y,x)$ for all $x,y \in \Lp$, is called a \emph{similarity valuation} over $\Lp$ if:
\begin{itemize}
    \item $0 \leq \rho(x,y) \leq \min\{\rk(x),\rk(y)\}$ for all $x,y \in \Lp$,
    \item $\rho(x,y) \leq \rho(z,y)$ for all $x \preceq z$,
    \item $\rho(x,y) = \rk(x)$ if and only if $x \preceq y$.
\end{itemize}
\label{defn:rho_prop}
\end{defn}

\begin{remark}
    The term `valuation' is often used in the order-theory literature \cite{stanley_2011} to denote functions on posets that respect the underlying partial order structure, and we use it in our context for the same reason.
\end{remark}
In the sequel, we describe similarity valuations for the various model posets discussed previously.  The conditions above make similarity valuations well-suited for quantifying the amount of discovery in an estimated model with respect to a true model.  The first condition states that the amount of discovery must be bounded above by the complexities of the true and estimated models (which are specified by the rank function).  The second condition requires similarity valuations to respect partial order structure so that more complex models do not yield less discovery than less complex ones.  The final condition expresses the desirable property that the amount of discovery contained in an estimated model is equal to the complexity of that model if and only if it is `contained in' the true model.  With these properties, we obtain the following analogs of true and false discoveries and of related quantities such as false discovery proportion.

\begin{defn}[true and false discoveries] Let $(\Lp,\preceq,\rk(\cdot))$ be a graded poset and let $\rho$ be a similarity valuation on $\Lp$.  Letting $x^\star \in \Lp$ be a true model and $\hat{x} \in \Lp$ be an estimate, the \emph{true discovery}, the \emph{false discovery}, and the \emph{false discovery proportion} are, respectively, defined as follows:
\begin{equation*}
\begin{aligned}
&\TD(\hat{x}, x^\star) := \rho(\hat{x},x^\star), \\ &\FD(\hat{x},x^\star) := \rk(\hat{x})-\rho(\hat{x},x^\star) = \rk(\hat{x}) - \TD(\hat{x}, x^\star), \\ &\mathrm{FDP}(\hat{x}, x^\star) := \frac{\rk(\hat{x})-\rho(\hat{x},x^\star)}{\rk(\hat{x})} = \frac{\FD(\hat{x},x^\star)}{\rk(\hat{x})}.
\end{aligned}
\end{equation*}
\end{defn}
With these definitions, we articulate our model selection objective more precisely:

\textbf{Goal}: \emph{identify the largest rank model subject to control in expectation or in probability on false discovery (proportion).}

This objective is akin to seeking the largest amount of discovery subject to control on false discovery (rate).  The data available to carry out model selection vary across our examples; in Section~\ref{sec:false_discovery_control} we describe methods to obtain false discovery control guarantees in various settings.

To carry out this program, a central question is the choice of a suitable similarity valuation for a graded model poset.  Indeed, it is unclear whether there always exists a similarity valuation for any graded model poset $(\Lp, \preceq, \rk(\cdot))$.  To address this question, consider the following function for $x,y \in \Lp$:
\begin{equation}
    \rho_{\text{meet}}(x,y) := \max_{z \preceq x, z \preceq y} ~ \rk(z). \label{eq:rho_meet}
\end{equation}
\begin{remark}
    In order theory, a poset $(\Lp, \preceq)$ is said to possess a \emph{meet} if for each $x,y \in \Lp$ there exists a $z \in \Lp$ satisfying $(i)$ $z \preceq x, z \preceq y$ and $(ii)$ for any $w \in \Lp$ with $w \preceq x, w \preceq y$, we have $w \preceq z$; such a $z$ is called the meet of $x,y$ and posets that possess a meet are called \emph{meet semi-lattices}.  Except for the poset in Example~\ref{ex:causal-learning} on causal structure learning, the posets in the other examples are meet semi-lattices (see Appendix~\ref{sec:meet_semi}).  The subscript `meet' in \eqref{eq:rho_meet} signifies that $\rho_{\text{meet}}$ is the rank of the meet for meet semi-lattices, although $\rho_{\text{meet}}$ is well-defined even if $(\Lp, \preceq)$ is not a meet semi-lattice.
\end{remark}
One can check that $\rho_{\text{meet}}$ is a similarity valuation on any graded poset $(\Lp, \preceq, \mathrm{rank}(\cdot))$; see Appendix \ref{sec:meet_rho} for a proof. For Example~\ref{ex:variable-selection} on variable selection, $\rho_{\text{meet}}$ has the desirable property that it reduces to the number of common variables in two models; thus, the general model selection goal formulated above reduces to the usual problem of maximizing the number of selected variables subject to control on the number of selected variables that are null.  Next, we describe the model selection problems we obtain in Examples~\ref{ex:clustering}-\ref{ex:partial-ranking} with $\rho_{\text{meet}}$ as the choice of similarity valuation.

In Example~\ref{ex:clustering} on clustering, the value of $\rho_{\text{meet}}$ for two partitions of $p$ variables is equal to $p$ minus the number of groups in the coarsest common refinement of the partitions.  The model selection problem is that of partitioning the variables into the smallest number of groups subject to control on the additional number of groups in the coarsest common refinement of the estimated and true partitions compared to the number of groups in the estimated partition.

Recall that the poset in Example~\ref{ex:multisample-testing} on multisample testing is the reverse of the poset in Example~\ref{ex:clustering}; thus, many of the notions from the preceding paragraph are appropriately `reversed' in Example~\ref{ex:multisample-testing}.  In particular, the value of $\rho_{\text{meet}}$ in Example~\ref{ex:multisample-testing} for two partitions of $p$ samples is equal to the number of groups in the finest common coarsening of the partitions.  The model selection problem entails partitioning the samples into the largest number of groups subject to control on the additional number of groups in the estimated partition compared to the number of groups in the finest common coarsening of the estimated and true partitions.

In Example~\ref{ex:causal-learning} on causal structure learning, the value of $\rho_{\text{meet}}$ for two CPDAGs $\mathcal{C}^{(1)}, \mathcal{C}^{(2)}$ is equal to the maximum number of edges in a CPDAG that encodes all the conditional independencies of $\mathcal{C}^{(1)}$ and of $\mathcal{C}^{(2)}$.  The model selection task is then to identify the CPDAG with the largest number of edges subject to control on the additional number of edges in the estimated CPDAG compared to the densest CPDAG that encodes all the conditional {independence} relationships in both the true and estimated CPDAGs.

In Example~\ref{ex:multiple-changepoint} on multiple changepoint estimation, suppose $x, y \in \{0,\dots,T\}^p$ are vectors of time indices specifying changepoints in $p$ signals.  We have that $\rho_{\text{meet}}(x,y) \allowbreak = p \cdot T - \allowbreak \sum_{i=1}^p \max\{x_i,y_i\}$.  The model selection problem entails identifying changes as quickly as possible subject to control on early detection of changes (i.e., declaring changes before they occur); this is a multivariate generalization of the classic quickest change detection problem \cite{Lorden1971PROCEDURESFR}.

In Example~\ref{ex:partial-ranking} on partial ranking, the value of $\rho_{\text{meet}}$ for two partial rankings is equal to the cardinality of the intersection of the associated relations, i.e., the number of common comparisons in the two partial rankings.  The associated model selection problem is that of identifying a partial ranking with the largest number of comparisons (i.e., the associated relation must have large cardinality) subject to control on the number of comparisons in the estimated partial ranking that are not in the true partial ranking.

In Examples~\ref{ex:variable-selection}-\ref{ex:partial-ranking}, the function $\rho_{\text{meet}}$ of \eqref{eq:rho_meet} provides a convenient way to assess the amount of discovery in an estimated model with respect to a true model, thereby yielding natural formulations for model selection.  However, in Examples~\ref{ex:complete-ranking}-\ref{ex:blind-source-separation}, $\rho_{\text{meet}}$ has some undesirable features.

Consider first the setup in Example~\ref{ex:complete-ranking} {on total ranking} for the set $S = \{a,b,c\}$ with the null model given by the ranking $\pi_{\text{null}}(a) = 1, \pi_{\text{null}}(b) = 2, \pi_{\text{null}}(b) = 3$, the true model given by the ranking $\pi^\star(a) = 3, \pi^\star(b) = 1, \pi^\star(c) = 2$ (Hasse diagram shown in Figure~\ref{fig:1}), and the estimated ranking given by $\hat{\pi}(a) = 2, \hat{\pi}(b) = 3, \hat{\pi}(c) = 1$.  In this case, one can see from Figure~\ref{fig:1} that $\rho_{\text{meet}}(\hat{\pi},\pi^\star) = 0$, which suggests that no discovery is made.  On the other hand, the inversion sets of these rankings are given by $\mathrm{inv}(\pi^\star; \pi_{\text{null}}) = \{(a,b), (a,c)\}$ and $\mathrm{inv}(\hat{\pi}; \pi_{\text{null}}) = \{(a,c), (b,c)\}$, and the element $(a,c)$ is common to both inversion sets as the fact that item $c$ is ranked higher than item $a$ in the true model has been discovered in the estimated model; this reasoning suggests that a positive quantity would be a more appropriate value for the similarity valuation between $\hat{\pi}$ and $\pi^\star$.  The key issue is that $\mathrm{inv}(\pi^\star; \pi_{\text{null}}) \cap \mathrm{inv}(\hat{\pi}; \pi_{\text{null}})$ is not an inversion set of any total ranking, but this intersection still carries valuable information about true discoveries made in $\hat{\pi}$ about $\pi^\star$.  However, the similarity valuation $\rho_{\text{meet}}$ only considers subsets of $\mathrm{inv}(\pi^\star; \pi_{\text{null}}) \cap \mathrm{inv}(\hat{\pi}; \pi_{\text{null}})$ that correspond to inversion sets of total rankings as the maximization in \eqref{eq:rho_meet} is constrained to be over poset elements.  Motivated by this discussion, we employ the following similarity valuation in Example~\ref{ex:complete-ranking} for total rankings $\pi, \tilde{\pi}$ (with respect to a null model $\pi_{\text{null}}$):
\begin{equation}
\rho_{\text{total-ranking}}(\pi,\tilde{\pi}) = |\mathrm{inv}(\pi; \pi_{\text{null}}) \cap \mathrm{inv}(\tilde{\pi}; \pi_{\text{null}})|.
\label{eqn:inv_set_ranking}
\end{equation}
With this similarity valuation, the model selection problem reduces to identifying a total ranking with the largest inversion set (with respect to $\pi_{\text{null}}$) subject to control on the number of comparisons in the inversion set of the estimated total ranking that are not in the inversion set of the true total ranking.

Next, in Example~\ref{ex:subspace-estimation}, $\rho_{\text{meet}}(\hat{x}, x^\star)$ is equal to the dimension of the intersection of the subspaces $\hat{x}, x^\star$.  When these subspaces have small dimensions, for example, $\rho_{\text{meet}}$ generically equals zero regardless of the angle between the subspaces; in words, $\rho_{\text{meet}}$ does not consider the smooth structure underlying the collection of subspaces.  As discussed in \cite{Taeb2020FD}, a more suitable measure of similarity is the sum of the squares of the cosines of the principal angles between the subspaces, which is expressed as follows using projection matrices onto subspaces $\mathcal{U}, \tilde{\mathcal{U}}$:
\begin{equation}
\rho_{\text{subspace}}(\mathcal{U},\tilde{\mathcal{U}}) = \mathrm{trace}(\mathcal{P}_{\mathcal{U}} \mathcal{P}_{\tilde{\mathcal{U}}}).
\label{eqn:subspace}
\end{equation}
The model selection task is to identify the largest-dimensional subspace subject to control on the sum of the squares of the cosines of the principal angles between the estimated subspace and the orthogonal complement of the true subspace.

Finally, $\rho_{\text{meet}}$ is inadequate as a similarity valuation in Example~\ref{ex:blind-source-separation} for the same reasons as in Example~\ref{ex:subspace-estimation} due to the underlying smooth structure, and we propose here a more appropriate alternative.  Given $B \in \R^{p \times k}, ~ \tilde{B} \in \R^{p \times \ell}$ (these matrices have unit-norm and linearly independent columns representing source signals), suppose without loss of generality that $k \leq \ell$ (due to the symmetry of similarity valuations) and let $\text{Perm}(\ell)$ be the collection of bijections on $\{1, \dots, \ell\}$.  With this notation, consider the following similarity valuation:
\begin{equation}
    \rho_{\text{source-separation}}(B,\tilde{B}) = \max_{\sigma \in \text{Perm}(\ell)} \sum_{i=1}^k (B^T \tilde{B})_{i,\sigma(i)}^2.
    \label{eqn:blind_source_rho}
\end{equation}

This valuation is better suited to quantify the degree of alignment between two collections of vectors in source separation than $\rho_{\text{meet}}$.  Model selection entails identifying the largest collection of source vectors subject to control on the difference in the number of estimated source vectors and the alignment between the true and estimated source vectors as evaluated by $\rho_{\text{source-separation}}$.

Table~\ref{summary_table} summarizes our discussion of the various model posets and their associated similarity valuations.  In conclusion, while $\rho_{\text{meet}}$ is a similarity valuation for any model poset, it is not always the most natural choice, and identifying a suitable similarity valuation that captures the essential features of an application is key to properly formulating a model selection problem.  This situation is not unlike the selection of an appropriate loss function in point estimation -- while there exist many candidates that are mathematically valid, the utility of an estimation procedure in the context of a problem domain depends critically on a well-chosen loss.

\begin{table*}
\centering
\scalebox{0.75}{
\begin{tabular}{ |c|c|c|c|c|c|c| } 
\hline
\textbf{\makecell{problem \\domain}} & \textbf{models} & \textbf{\makecell{least element \\(i.e. global null)}} &\textbf{\makecell{partial\\order}}& \textbf{\makecell{rank \\(i.e. model \\complexity)}}&\textbf{\makecell{similarity valuation\\(i.e. true discoveries)}}  \\
\hline
\emph{\makecell{variable \\selection}}& \makecell{subsets of \\ $\{1,\dots,p\}$}& $\emptyset$& \makecell{inclusion of \\ subsets}&  \makecell{cardinality of \\subset}& \makecell{subsets $x,\tilde{x}$; \\$\rho(x,\tilde{x}) = |x\cap \tilde{x}|$}\\
\hline
\emph{clustering} &\makecell{partitions of \\$\{1,\dots,p\}$} & $\{1\},\{2\},\dots,\{p\}$& \makecell{refinement \\ of partition}& $p - \# \text{groups}$& \makecell{partitions $x,\tilde{x}$;\\ $\rho(x,\tilde{x}) = p - \#$ groups in\\ coarsest common refinement}\\
\hline
\emph{\makecell{multisample \\testing}} &\makecell{partitions of \\$\{1,\dots,p\}$} & $\{1,2,\dots,p\}$ & \makecell{coarsening \\ of partition}&  $\# \text{groups}$ - 1 & \makecell{partitions $x,\tilde{x}$; \\ $\rho(x,\tilde{x}) = \# \text{ groups in}$\\ finest common coarsening} \\
\hline
\emph{\makecell{causal\\ structure \\ learning}}& \makecell{completed partially \\ directed acyclic graphs \\ (CPDAG) \\ on a set of variables} &\makecell{CPDAG \\ with no \\ edges} &\makecell{inclusion \\ of conditional \\ dependencies \\ encoded by \\ CPDAGs} & $\# \text{edges}$& \makecell{$\text{CPDAGs } \mathcal{C},\tilde{\mathcal{C}}$; \\ $\rho(\mathcal{C},\tilde{\mathcal{C}}) = \#\text{edges in }$\\
densest CPDAG encoding\\ conditional {independencies} \\of both $\mathcal{C},\tilde{\mathcal{C}}$}\\
\hline
\emph{\makecell{multiple \\changepoint}} & \makecell{elements of \\ $\{0,\dots,T\}^p$} & $(T,T,\dots,T)$ & \makecell{entrywise\\ reverse \\ ordering} & \makecell{$p\cdot T$ minus \\sum of entries} & \makecell{changepoint vectors $x,\tilde{x}$;\\ $\rho(x,\tilde{x}) = p\cdot T-\sum_i \max\{x_i,\tilde{x}_i\}$}\\
\hline
\emph{\makecell{partial \\ranking}} & \makecell{relations specified by \\ strict partial orders\\ on a set of items} & $\emptyset$ &\makecell{inclusion of \\ sets specifying\\relations} & \makecell{cardinality of \\ set specifying \\ relation} & \makecell{sets $\mathcal{R},\tilde{\mathcal{R}}$ \\ specifying relations; \\ $\rho(\mathcal{R},\tilde{\mathcal{R}}) = |\mathcal{R} \cap \tilde{\mathcal{R}}|$}  \\
\hline
\emph{\makecell{total\\ ranking}} & \makecell{total orders \\ on a set of items}&\makecell{base ranking\\ $\pi_\text{null}$}&\makecell{inclusion of \\inversion sets \\ w.r.t. $\pi_{\text{null}}$} &  \makecell{cardinality of \\ inversion set \\ w.r.t. $\pi_{\text{null}}$}& \makecell{total orders $\pi,\tilde{\pi}$; \\$\rho(\pi,\tilde{\pi}) = |\mathrm{inv}(\pi;{\pi}_\text{null}) \cap \mathrm{inv}(\tilde{\pi};{\pi}_\text{null})|$} \\
\hline
\emph{\makecell{subspace\\ estimation}} & subspaces in $\mathbb{R}^p$ & $\{0\}$ & \makecell{inclusion of \\subspaces} & \makecell{dimension of \\ subspace}& \makecell{subspaces $\mathcal{U},\tilde{\mathcal{U}}$;\\$\rho(\mathcal{U},\tilde{\mathcal{U}}) = \mathrm{trace}(\Proj_\mathcal{U}\Proj_{\tilde{\mathcal{U}}})$} \\
\hline
\emph{\makecell{blind source\\ separation}} & \makecell{linearly independent\\ subsets of $\mathbb{R}^p$} & $\emptyset$ & \makecell{inclusion of \\subsets}& \makecell{cardinality of \\subset}& \makecell{subsets given by columns of \\ $B \in \mathbb{R}^{p \times k},\tilde{B} \in \mathbb{R}^{p \times \ell}, k\leq \ell$;\\ $\rho(B,\tilde{B})= \max_{\sigma \in \text{Perm}(\ell)} \sum_{i=1}^k (B^T \tilde{B})_{i,\sigma(i)}^2$} \\
\hline
\end{tabular}}
\caption{Problem classes and associated characterization of model selection via posets.}
\label{summary_table}
\end{table*}
\section{False Discovery Control over Posets}
\label{sec:false_discovery_control}
In this section, we turn our attention to the task of identifying models of large rank that provide false discovery control.  We begin in Section~\ref{sec:greedy_approaches} with a general greedy strategy for poset search that facilitates the design of model selection procedures, and we specialize this framework to specific approaches in Sections~\ref{sec:stable} and \ref{sec:testing}.  Some of the discussion in Section~\ref{sec:greedy_approaches} is relevant for all of the posets in Examples~\ref{ex:variable-selection}-\ref{ex:blind-source-separation}, while the methodology presented in Sections~\ref{sec:stable}-\ref{sec:testing} is applicable to general discrete posets with integer-valued similarity valuations such as in Examples~\ref{ex:variable-selection}-\ref{ex:complete-ranking}.  Along the way, we remark on some of the challenges that arise in the two continuous cases of Examples~\ref{ex:subspace-estimation}-\ref{ex:blind-source-separation}.

\subsection{Greedy Approaches to Model Selection}
\label{sec:greedy_approaches}
To make progress on the problem of identifying large rank models that provide control on false discovery, we begin by noting that the false discovery $\FD(\hat{x},x^\star)$ in an estimated model $\hat{x}$ with respect to a true model $x^\star$ may be expressed as the following telescoping sum for any path $(x_0, x_1, \dots, x_{k-1}, x_k)$ with $x_0$ being the least element $x_{\text{least}}$ and $x_k = \hat{x}$:
\begin{equation}
     \FD(\hat{x},x^\star) = \sum_{i=1}^{k} 1-[\rho(x_i,\xstar)-\rho(x_{i-1},\xstar)].
     \label{eqn:telescoping_sum}
 \end{equation}
The term $1-[\rho(x_i,\xstar)-\rho(x_{i-1},\xstar)]$ may be interpreted as the ``additional false discovery'' incurred by the model $x_i$ relative to the model $x_{i-1}$.  The above decomposition of false discovery in terms of a path from the least element to an estimated model suggests a natural approach for model selection.  In particular, we observe that a sufficient condition for $\FD(\hat{x},x^\star)$ to be small is for each term in the above sum to be small.  Thus, we will greedily grow a path starting from the least element $x_0 = x_{\text{least}}$ by adding one element $x_i$ at a time such that each $(x_{i-1},x_i)$ is a covering pair and each $1-[\rho(x_i,\xstar)-\rho(x_{i-1},\xstar)]$ is small.  We continue this process until we can no longer guarantee that $1-[\rho(x_i,\xstar)-\rho(x_{i-1},\xstar)]$ is small.

For such a procedure to be fruitful, we require some data-driven method to bound $1-[\rho(x_i,\xstar)-\rho(x_{i-1},\xstar)]$ as the true model $\xstar$ is not known.  Our objective, therefore, is to design a data-dependent function $\Psi:\{(a,b) ~|~ b \text{ covers } a \text{ in } \Lp\} \to [0,1]$ that takes as input covering pairs and outputs a number in the interval $[0,1]$, and further satisfies the property that $\Psi(u,v)$ being small is a sufficient condition for $1-[\rho(v,\xstar)-\rho(u,\xstar)]$ to be small (in expectation or in probability).  Given such a function, we grow a path using the greedy strategy outlined above by identifying at each step a covering pair that minimizes $\Psi$.  Algorithm~\ref{alg:poset_stability_discrete} provides the details.  In Sections~\ref{sec:stable} and \ref{sec:testing}, we present two approaches for designing suitable functions $\Psi$: one based on a notion of stability and the other based on testing.  Proofs that both these methods control for false discoveries are presented in Section~\ref{sec:proofs}.

\begin{algorithm}
\caption{Greedy sequential algorithm for model selection} 
\begin{algorithmic}[1]
\State {\bf Input}: poset $\Lp$, threshold $\alpha \in [0,1]$; data-dependent function $\Psi:\{(a,b) ~|~ b \text{ covers } a \text{ in } \Lp\} \to [0,1]$
\State {\bf Greedy selection}: Set $u= x_{\text{least}}$ and perform:
\begin{itemize}
    \item [(a)] find $v_\text{opt} \in \argmin_{\{(u,v) ~|~ v \text{ covers } u \text{ in } \Lp\}} {\Psi}(u,v)$. 
    \item [(b)] if ${\Psi}(u,v_\text{opt}) \leq \alpha$, set $u = v_\text{opt}$ and repeat steps (2a-2b). Otherwise, stop.
\end{itemize}
\State {\bf Output}: return $\hat{x} = u$
\end{algorithmic}
\label{alg:poset_stability_discrete}
\end{algorithm}

In designing a suitable function $\Psi$ so that $1-(\rho(v,x^\star)-\rho(u,x^\star))$ is small (in expectation or in probability) whenever $\Psi(u,v)$ is small, we note that the examples presented in Section~\ref{sec:framework} exhibit an important invariance.  Specifically, in each example there are distinct covering pairs $(u,v)$ and $(u',v')$ such that $1-[\rho(v,x^\star)-\rho(u,x^\star)] = 1-[\rho(v',x^\star)-\rho(u',x^\star)]$ for every true model $x^\star$.  Accordingly, it is natural that the function $\Psi$ also satisfies the property that $\Psi(u,v) = \Psi(u',v')$; stated differently, one need only specify $\Psi$ for a `minimal' set of covering pairs.  We present next a definition that formalizes this notion precisely.

\begin{defn}[Minimal covering pairs] Consider a graded poset $(\Lp, \preceq, \rk(\cdot))$ endowed with a similarity valuation $\rho$.  A subset $\mathcal{S} \subset \{(a,b) ~|~ b \text{ covers } a \text{ in } \Lp\}$ of covering pairs in $\Lp$ is called \emph{minimal} if the following two properties hold:
\begin{itemize}
\item For each covering pair $(u', v') \notin \mathcal{S}$, there exists $(u,v) \in \mathcal{S}$ with $\rk(v) \leq \rk(v')$ such that $\rho(v,z)-\rho(u,z) = \rho(v',z)-\rho(u',z)$ for all $z \in \Lp$.
\item For distinct covering pairs $(u,v), (u',v') \in \mathcal{S}$, there exists some $z \in \Lp$ such that $\rho(v,z)-\rho(u,z) \neq \rho(v',z)-\rho(u',z)$.
\end{itemize}
\label{defn:set_S}
\end{defn}
In words, a minimal set of covering pairs $\mathcal{S}$ for a graded poset $\Lp$ is an inclusion-minimal collection of smallest rank covering pairs for which it suffices to consider the values of $\Psi$.  For Example~\ref{ex:variable-selection} on variable selection with the similarity valuation $\rho_{\text{meet}}$, a minimal set of covering pairs is given by $\mathcal{S} = \{(\emptyset, \{i\}) ~|~ i = 1,\dots,p\}$ and this minimal set is unique.  In general, however, such sets are not unique; see Appendix~\ref{sec:setS} where we derive minimal sets of covering pairs for several examples.  Minimal sets of covering pairs are significant methodologically from both computational and statistical perspectives.  In particular, several of our bounds for discrete posets depend on the cardinality $|\mathcal{S}|$ and these also involve computations that scale in number of operations with $|\mathcal{S}|$.  Therefore, identifying a minimal set of covering pairs that is small in cardinality is central to the success of our proposed methods.  In the remainder of this section, we assume that a minimal set of covering pairs $\setS$ for a given model poset $\Lp$ is available.

\subsection{Model Selection via Stability}
\label{sec:stable}

Our first method for designing a suitable function $\Psi$ to employ in Algorithm~\ref{alg:poset_stability_discrete} is based on subsampling and corresponding model averaging.  We assume that we have access to a base procedure $\hat{x}_{\mathrm{base}}$ that provides model estimates from data as well as a dataset $\mathcal{D}$ consisting of observations drawn from a probability distribution parameterized by the true model $x^\star$, and our approach is to aggregate the model estimates provided by $\hat{x}_{\mathrm{base}}$ on subsamples of $\mathcal{D}$.  The requirements on the quality of the procedure $\hat{x}_{\mathrm{base}}$ are quite mild, and we prove bounds in the sequel on the false discovery associated with the aggregated model.  In particular, the aggregation method ensures that the averaged model is `stable' in the sense that it contains discoveries that are supported by a large fraction of the subsamples.  Our method generalizes the stability selection method for variable selection \cite{Meins2010Stability,Shah2013Stability} and subspace stability selection for subspace estimation \cite{Taeb2020FD}.  We demonstrate the broad applicability of this methodology in Section~\ref{sec:experiments} by applying it to several examples from Section~\ref{sec:framework}.

Formally, fix a positive even integer $B$ and obtain $B/2$ complementary partitions of the dataset $\mathcal{D}$, each of which partitions $\mathcal{D}$ into two subsamples of equal size.  Let this collection of subsamples be denoted $\{\mathcal{D}^{(\ell)}\}_{\ell=1}^B$, and let $\hat{x}_{\mathrm{base}}(\mathcal{D}^{(\ell)})$ denote the model estimate obtained by applying the base procedure to the subsample $\mathcal{D}^{(\ell)}$.  For any covering pair $(u,v)$ of a model poset $\Lp$, we define:
\begin{equation}
    \Psi_{\mathrm{stable}}(u,v) := 1 - \frac{1}{B}\sum_{\ell=1}^B \frac{\rho(v,\hat{x}_{\mathrm{base}}(\mathcal{D}^{(\ell)}))-\rho(u,\hat{x}_{\mathrm{base}}(\mathcal{D}^{(\ell)}))}{c_{\Lp}(u,v)},
    \label{eqn:psi_stable}
\end{equation}
where $c_{\Lp}(u,v) := \max_{z \in \Lp}\rho(v,z)-\rho(u,z)$.  Appealing to properties of similarity valuations, we have that $\rho(v,\hat{x}_{\mathrm{base}}(\mathcal{D}^{(\ell)}))-\rho(u,\hat{x}_{\mathrm{base}}(\mathcal{D}^{(\ell)})) \geq 0$ and $c_{\Lp}(u,v) \geq 1$.  The term $\rho(v,\hat{x}_{\mathrm{base}}(\mathcal{D}^{(\ell)}))-\rho(u,\hat{x}_{\mathrm{base}}(\mathcal{D}^{(\ell)}))$ measures the additional discovery about $\hat{x}_{\mathrm{base}}(\mathcal{D}^{(\ell)})$ in the model $v$ relative to the model $u$, while the quantity $c_{\Lp}(u,v)$ serves as normalization to ensure that $\Psi_{\mathrm{stable}}(u,v) \in [0,1]$.  In particular, $\Psi_{\mathrm{stable}}(u,v)$ being small implies that the additional discovery represented by the model $v$ over the model $u$ is supported by a large fraction of the subsamples $\{\mathcal{D}^{(\ell)}\}_{\ell=1}^B$.  Consequently, when $\Psi_{\mathrm{stable}}$ is employed in the context of Algorithm~\ref{alg:poset_stability_discrete} in which we greedily grow a path, each `step' in the path corresponds to a discovery that is supported by a large fraction of the subsamples.  We provide theoretical support for this approach in Theorem~\ref{thm:main_stability} in the sequel and the proof proceeds by showing that $\Psi_{\mathrm{stable}}(u,v)$ being small implies that $\mathbb{E}[1-(\rho(u,x^\star)-\rho(v,x^\star))]$ is small; we combine this observation with the telescoping sum formula \eqref{eqn:telescoping_sum} to obtain a bound on the expected false discovery of the model estimated by Algorithm~\ref{alg:poset_stability_discrete}.

When Algorithm~\ref{alg:poset_stability_discrete} with $\Psi = \Psi_{\mathrm{stable}}$ is specialized to Example~\ref{ex:variable-selection} and Example~\ref{ex:subspace-estimation}, we obtain the stability selection procedure of \cite{Meins2010Stability,Shah2013Stability} and the subspace stability selection method of \cite{Taeb2020FD}.  For variable selection in particular, Algorithm~\ref{alg:poset_stability_discrete} with $\Psi = \Psi_{\mathrm{stable}}$ outputs the subset of variables that appear in at least a $1 - \alpha$ fraction of the models estimated by the base procedure when applied to the subsamples $\{\mathcal{D}^{(\ell)}\}_{\ell=1}^B$.  More generally, Algorithm~\ref{alg:poset_stability_discrete} with $\Psi = \Psi_{\mathrm{stable}}$ also provides a procedure for model selection in Examples~\ref{ex:clustering}-\ref{ex:complete-ranking} corresponding to discrete model posets.

\begin{theorem}[false discovery control for Algorithm~\ref{alg:poset_stability_discrete} with $\Psi = \Psi_\mathrm{stable}$] Let $(\Lp, \preceq, \rk(\cdot))$ be a graded discrete model poset with integer-valued similarity valuation $\rho$ and let $\mathcal{S}$ be an associated set of minimal covering pairs.  Let $\hat{x}_{\mathrm{base}}$ be a base estimator.  Suppose the dataset $\mathcal{D}$ employed in the computation of $\Psi_{\text{stable}}$ consists of i.i.d. observations from a distribution parametrized by the true model $x^\star \in \Lp$, and suppose $\hat{x}_{\mathrm{sub}}$ is an estimator obtained by applying $\hat{x}_{\mathrm{base}}$ to a subsample of $\mathcal{D}$ of size $|\mathcal{D}| / 2$.  Fix $\alpha \in (0,1/2)$ and a positive, even integer $B$.  The output $\hat{x}_{\text{stable}}$ from Algorithm~\ref{alg:poset_stability_discrete} with $\Psi = \Psi_{\mathrm{stable}}$ satisfies the following false discovery bound
\begin{equation}
\E[\FD(\hat{x}_{\mathrm{stable}},\xstar)] \leq \hspace{-0.1in}\sum_{\substack{(u,v)\in\setS \cap \mathcal{T}_\text{null}}} \hspace{-0.2in}\frac{\E[\rho(v,\hatx_\mathrm{sub})-\rho(u,\hatx_\mathrm{sub})]^2}{(1-2\alpha)c_{\Lp}(u,v)^2}.
\label{eqn:bound_stab}
\end{equation}
Here the set $\mathcal{T}_\text{null}:=\{(u,v) \text{ covering pair in } \Lp ~|~ \rho(v,x^\star) = \rho(u,x^\star)\}$ consists of all covering pairs $(u,v)$ for which there is no additional discovery in the model $v$ over the model $u$ with respect to the true model $x^\star$.
\label{thm:main_stability}
\end{theorem}
In the bound \eqref{eqn:bound_stab}, the numerator $\E[\rho(v,\hatx_{\mathrm{sub}})-\rho(u,\hatx_{\mathrm{sub}})]^2$ of each summand characterizes the quality of the base estimator on subsamples; base estimators for which this term is small, when employed in the computation of $\Psi_{\text{stable}}$ in Algorithm~\ref{alg:poset_stability_discrete}, yield models $\hat{x}_{\text{stable}}$ with small false discovery.

\begin{remark} When specialized to Example~\ref{ex:variable-selection} on variable selection with similarity valuation $\rho_{\text{meet}}$, we recover Theorem 1 of \cite{Shah2013Stability}. Specifically, in \eqref{eqn:bound_stab}, we have that $c_{\Lp}(u,v)= 1$ for any covering pair $(u,v)$ and $\sum_{(u,v)\in\setS\cap\mathcal{T}_\text{null}} \E[\rho(v,\hatx_{\mathrm{sub}})-\rho(u,\hatx_{\mathrm{sub}})]^2 = \sum_{\mathrm{null }~ i } \mathbb{P}[\text{variable }i \text{ selected by }\hat{x}_{\mathrm{sub}}]^2$. 
\end{remark}

Theorem~\ref{thm:main_stability} is general in its applicability to all the discrete posets in Section~\ref{sec:framework}, and it provides an intuitive bound on expected false discovery.  Nonetheless, it requires a characterization of the quality of the base estimator $\hat{x}_{\mathrm{base}}$ employed on subsamples.  When such a characterization is unavailable, the false discovery bound \eqref{eqn:bound_stab} may not be easily computable in practice.  To address this shortcoming and obtain easily computable bounds on false discovery, we consider natural assumptions on the estimator $\hat{x}_{\mathrm{sub}}$ corresponding to the base estimator $\hat{x}_{\mathrm{base}}$ applied to subsamples; these assumptions generalize those developed in \cite{Meins2010Stability,Taeb2020FD} for stability-based methods for variable selection and subspace estimation.  To formulate these assumptions, we introduce some notation. Let $\rk(\Lp) := \max_{u \in \Lp} \rk(u)$ be the largest rank of an element in $\Lp$ and let $\setS_k := \{(u,v) \in \setS~|~ \rk(v) = k\}$ for each  $k \in [\rk(\Lp)]$. 

\begin{assumption}[better than random guessing]For each $k \in [\rk(\Lp)]$ with $\setS_k \neq \emptyset$, we have that
\begin{equation*}
\begin{aligned}
    &\sum_{\substack{(u,v)\in \mathcal{S}_k \cap \mathcal{T}_\mathrm{null}}} \frac{1}{|\mathcal{S}_k \cap \mathcal{T}_\mathrm{null}|} \cdot \frac{\mathbb{E}[\rho(v,\hatx_{\mathrm{sub}})-\rho(u,\hatx_{\mathrm{sub}})]}{c_{\Lp}(u,v)} \\ &\leq \sum_{\substack{(u,v)\in \mathcal{S}_k\setminus \mathcal{T}_\mathrm{null}}} \frac{1}{|\mathcal{S}_k \setminus \mathcal{T}_\mathrm{null}|} \frac{\mathbb{E}[\rho(v,\hatx_{\mathrm{sub}})-\rho(u,\hatx_{\mathrm{sub}})]}{c_{\Lp}(u,v)}.
    \end{aligned}
\end{equation*}
\label{assump:better_than_rg}
\end{assumption}

\begin{assumption}[invariance in mean] For each $k \in [\rk(\Lp)]$ with $\setS_k \neq \emptyset$, we have that $\frac{\mathbb{E}[\rho(v,\hatx_{\mathrm{sub}})-\rho(u,\hatx_{\mathrm{sub}})]}{c_{\Lp}(u,v)}$ is the same for each $(u,v)\in \mathcal{S}_k \cap \mathcal{T}_\mathrm{null}$.
\label{assump:exchangeability}
\end{assumption}

In words, Assumption~\ref{assump:better_than_rg} states that the average normalized difference in similarity valuation of the estimator $\hatx_{\mathrm{sub}}$ is smaller over `null' covering pairs than over non-null covering pairs. Assumption~\ref{assump:exchangeability} states that the expected value of the normalized difference in similarity of $\hatx_{\mathrm{sub}}$ is the same for each `null' covering pair.  For the case of variable selection (Example~\ref{ex:variable-selection}), Assumption~\ref{assump:better_than_rg} reduces precisely
to the `better than random guessing' assumption employed by \cite{Meins2010Stability}, namely that the expected number of true positives divided by the expected number of false positives selected by the estimator $\hatx_{\mathrm{sub}}$ is larger than the same ratio for an estimator that selects variables at random. As a second condition, \cite{Meins2010Stability} required that the random variables in the collection $\{\mathbb{I}[i \in \hatx_{\mathrm{sub}}]: i \text{ null}\}$ are exchangeable. Our Assumption~\ref{assump:exchangeability} when specialized to variable selection reduces to the requirement that each of the random variables in the collection $\{\mathbb{I}[i \in \hatx_{\mathrm{sub}}]: i \text{ null}\}$ has the same mean. As a second illustration, consider the case of total ranking (Example~\ref{ex:complete-ranking}) involving items $a_1,\dots,a_p$, with the least element $\pi_\text{null}$ given by $\pi_\text{null}(a_i) = i, ~ i=1,\dots,p$, the true total ranking by $\pi^\star$, and the estimator on subsamples by $\hat{\pi}_{\mathrm{sub}}$. Fix any $k \in \{1,\dots,p-1\}$. 
 Assumption 1 states that the expected number of pairs $(a_i,a_j) \in  \text{inv}(\hat{\pi}_\mathrm{sub};\pi_\text{null}) \cap \text{inv}(\pi^\star;\pi_\text{null})$ with $j-i = k$ divided by the expected number of pairs $(a_i,a_j) \in \text{inv}(\hat{\pi}_\mathrm{sub};\pi_\text{null}) \setminus \text{inv}(\pi^\star;\pi_\text{null})$ with $j-i = k$ is larger than the same ratio for an estimator that outputs a total ranking at random. Assumption 2 states that the probability that $(a_i,a_j) \in \text{inv}(\hat{\pi}_\mathrm{sub};\pi_\text{null})$ is the same for all pairs $(a_i,a_j)$ with $j-i = k$ and $(a_i,a_j) \not\in \text{inv}(\pi^\star;\pi_\text{null})$. See Appendix~\ref{sec:assumptions_total_ranking} for a formal derivation.

\begin{theorem}[refined false discovery control for Algorithm~\ref{alg:poset_stability_discrete} with $\Psi = \Psi_{\mathrm{stable}}$] Consider the setup of Theorem~\ref{thm:main_stability}, and suppose additionally that Assumptions \ref{assump:better_than_rg} and \ref{assump:exchangeability} are satisfied. The output $\hat{x}_{\mathrm{stable}}$ from Algorithm~\ref{alg:poset_stability_discrete} with $\Psi = \Psi_{\mathrm{stable}}$ satisfies the false discovery bound:
\begin{equation}
\mathbb{E}[\FD(\hatx_{\mathrm{stable}},\xstar)] \leq \sum_{k \in[\rk(\Lp)], \setS_k \neq \emptyset}\frac{q_k^2}{|\setS_k|(1-2\alpha)},
\label{eqn:bound_refined}
\end{equation}
where $q_k = \sum_{(u,v) \in \setS_k} \E[\rho(v,\hatx_{\mathrm{sub}})-\rho(u,\hatx_{\mathrm{sub}})]/c_{\Lp}(u,v)$. 
\label{prop:refined_stability}
\end{theorem}
 The quantities in the bound \eqref{eqn:bound_refined} may be readily computed in practice. In particular, each $\setS_k$ and $c_\Lp(\cdot,\cdot)$ depends only on the model poset $\Lp$ and each $q_k$ can be approximated as $q_k \approx \frac{1}{B}\sum_{\ell=1}^B \sum_{(u,v)\in\setS_k}\frac{\rho(v,\hat{x}_{\mathrm{base}}(\mathcal{D}^{(\ell)}))-\rho(u,\hat{x}_{\mathrm{base}}(\mathcal{D}^{(\ell)}))}{c_{\Lp}(u,v)}$.  We give characterizations of the sets $\setS_k$ and $c_{\Lp}(\cdot,\cdot)$ for posets corresponding to total ranking, partial ranking, clustering, and causal structure learning in Appendix~\ref{sec:setS}.
\begin{remark}Specializing Theorem~\ref{prop:refined_stability} to the case of variable selection, we arrive at the bound in Theorem 1 of \cite{Meins2010Stability}. Specifically, note that for the Boolean poset with the similarity valuation $\rho_\text{meet}$, $\setS_k = \emptyset$ for $k \geq 2$, $|\setS_1| = \# \text{ variables}$, and $q_1 = \sum_{i} \mathbb{P}[\text{variable }i \text{ selected by }\hatx_{\mathrm{sub}}]$ is the average number of variables selected by the estimator $\hatx_{\mathrm{sub}}$.
\end{remark}

Turning our attention to Examples~\ref{ex:subspace-estimation}-\ref{ex:blind-source-separation}, the situation is considerably more complicated with continuous model posets.  A result for these two cases under the same setup as in Theorem~\ref{thm:main_stability} yields the following bound for $\alpha \in (0,1/2)$ (see Appendix~\ref{sec:proof_general_stability}):
\begin{equation}
\begin{aligned}
\E[\FD(\hatx_{\mathrm{stable}},\xstar)] ~ \leq ~ &\frac{2\alpha+2\sqrt{\alpha}}{1-\alpha}\mathbb{E}[\rk(\hatx_{\mathrm{sub}})] + \mathbb{E}[\sqrt{\FD(\hatx_{\mathrm{sub}},\xstar)}]^2.
\end{aligned}
\label{eqn:general_bound_overview}
\end{equation}
The first term in the bound is a function of the average number of discoveries made by the estimator $\hatx_{\mathrm{sub}}$, and this term is smaller for $\alpha \approx 0$.  The second term in the bound concerns the quality of the estimator $\hatx_{\mathrm{sub}}$.  Specifically, note that Jensen's inequality implies $\E[\sqrt{\FD(\hatx_{\mathrm{sub}}, x^\star)}]^2 \leq \E[\FD(\hatx_{\mathrm{sub}}, x^\star)]$, so that the improvement provided by the estimator $\hatx_{\mathrm{stable}}$ based on subsampling and model averaging over the estimator $\hatx_{\mathrm{sub}}$ that simply employs the base estimator on subsamples is characterized by $\text{var}(\FD(\hatx_{\mathrm{sub}}, x^\star))$.  Thus, the key remaining task as before is to characterize the properties of the estimator $\hatx_{\mathrm{sub}}$.  However, the difficulty with the continuous examples is that conditions akin to Assumptions~\ref{assump:better_than_rg}-\ref{assump:exchangeability} are substantially more challenging to formulate and analyze at an appropriate level of generality.  (One such effort under a limited setting for the case of subspace estimation is described in \cite{Taeb2020FD}.)  It is of interest to develop such a general framework for continuous model posets, and we leave this as a topic for future research.

\subsection{Model Selection via Testing}
\label{sec:testing}
Our second approach to designing a suitable function $\Psi$ to employ in Algorithm~\ref{alg:poset_stability_discrete} is based on testing the following null hypothesis for each (minimal) covering pair $(u,v)$ of a discrete model poset $\Lp$:
\begin{equation}
\begin{aligned}
    H_0^{u,v}&: \rho(v,x^\star)=\rho(u,x^\star), \\ \Psi_{\mathrm{test}}(u,v) &:= \text{p-value corresponding to } H_0^{u,v}.
    \end{aligned}
        \label{eqn:hypothesis}
\end{equation}
The null hypothesis $H_0^{u,v}$ in \eqref{eqn:hypothesis} states that there is no additional discovery about $x^\star$ in the model $v$ relative to the model $u$, and small values of $\Psi_\mathrm{test}(u,v)$ provide evidence for rejecting this null hypothesis and accepting the alternative that $\rho(v,x^\star) > \rho(u,x^\star)$.  When $\Psi_{\mathrm{test}}$ is employed in the context of Algorithm~\ref{alg:poset_stability_discrete} in which we greedily grow a path, each `step' in the path corresponds to a discovery for which we have the `strongest evidence' using the above test.  Our next result provides theoretical support for this method.

\begin{theorem}[false discovery control for Algorithm~\ref{alg:poset_stability_discrete} with $\Psi = \Psi_{\mathrm{test}}$] Let $(\Lp, \preceq, \rk(\cdot))$ be a graded discrete model poset with integer-valued similarity valuation $\rho$ and let $\mathcal{S}$ be an associated set of minimal covering pairs. The output $\hat{x}_{\mathrm{test}}$ of Algorithm~\ref{alg:poset_stability_discrete} with $\Psi = \Psi_{\text{test}}$ satisfies the false discovery bound $\mathbb{P}\left(\FD(\hat{x}_\mathrm{test},x^\star)>0\right) \leq \alpha|\mathcal{S}|$.
\label{thm:testing_general}
\end{theorem}
The multiplicity factor involving the cardinality of the set of minimal covering pairs $\mathcal{S}$ is akin to a Bonferroni-type correction, and it highlights the significance of identifying a set of minimal covering pairs of small cardinality.  We emphasize that although Algorithm~\ref{alg:poset_stability_discrete} with $\Psi = \Psi_{\text{test}}$ proceeds via sequential hypothesis testing, the procedure is applicable to general model classes with no underlying Boolean logical structure; in particular, it is the graded poset structure underlying our framework that facilitates such methodology.

As an illustration of the multiplicity factor $|\mathcal{S}|$ for different settings, we have that $|\setS| = p(p-1)$ for partial ranking; $|\mathcal{S}| = \sum_{k=1}^{p-1}\big(\begin{smallmatrix}p\\k+1 \end{smallmatrix}\big)\sum_{\ell=1}^k \big(\begin{smallmatrix}k+1\\\ell\end{smallmatrix}\big)$ for clustering; and $|\mathcal{S}| = \frac{p(p-1)}{2}$ for total ranking.  See Appendix~\ref{sec:setS} for further details.

The graded poset structure of a model class can also yield more powerful model selection procedures than those obtained by the greedy procedure of Algorithm~\ref{alg:poset_stability_discrete}.  We give one such illustration next in which a collection of model estimates that each exhibits zero false discovery (with high probability) can be `combined' to derive a more complex model that also exhibits zero false discovery.  Formally, a poset $(\Lp, \preceq)$ is said to possess a \emph{join} if for each $x,y \in \Lp$ there exists a $z \in \Lp$ satisfying $(i)$ $z \succeq x, z \succeq y$ and $(ii)$ for any $w \in \Lp$ with $w \succeq x, w \succeq y$, we have $w \succeq z$; such a $z$ is called the join of $x,y$ and posets that possess a join are called \emph{join semi-lattices} (these are dual to the notion of a meet defined in Section~\ref{sec:framework}).  Except for the posets in Examples~\ref{ex:causal-learning}, ~\ref{ex:partial-ranking}, and ~\ref{ex:blind-source-separation}, the posets in the other examples are join semi-lattices (see Appendix~\ref{sec:meet_semi}).  For a model class that is a join semi-lattice, suppose we are provided estimates $\hatx^{(1)}, \dots, \hatx^{(m)}$ of a true model $x^\star$ such that $\FD(\hatx^{(j)}, x^\star) = 0, ~ j=1,\dots,m$ (for example, by appealing to greedy methods such as Algorithm~\ref{alg:poset_stability_discrete} or its variants).  Appealing to the properties of a similarity valuation, we can conclude that the join $\hatx_{\text{join}}$ of $\hatx^{(1)},\dots,\hatx^{(m)}$ satisfies $\FD(\hatx_{\text{join}}, x^\star) = 0$; in general, $\rk(\hatx_{\text{join}})$ is larger than $\rk(\hatx^{(1)}),\dots,\rk(\hatx^{(m)})$, and therefore, this procedure is one way to obtain a more powerful model by combining less powerful ones while still retaining control on the amount of false discovery.  The following result formalizes matters.
\begin{proposition}[using joins to obtain more powerful models] Let $(\Lp, \preceq, \rk(\cdot))$ be a graded discrete model poset that is a join semi-lattice with integer-valued similarity valuation $\rho$ and let $\mathcal{S}$ be an associated set of minimal covering pairs.  Consider a collection of estimates $\hat{x}^{(1)},\dots,\hat{x}^{(m)}$ of a true model $x^\star$ and let $\hatx_{\mathrm{join}}$ denote the join of $\hat{x}^{(1)},\dots,\hat{x}^{(m)}$.  Suppose for each $\hat{x}^{(j)}, ~ j=1,\dots,m$ there is a path from the least element of $\Lp$ to $\hat{x}^{(j)}$ such that every covering pair $(u,v)$ along the path satisfies $\Psi_\mathrm{test}(u,v) \leq \alpha$.  Then we have the false discovery bound $\mathbb{P}(\FD(\hatx_{\mathrm{join}},x^\star) > 0) \leq \alpha|\mathcal{S}|$.
\label{prop:join}
\end{proposition}

\section{Experiments}
\label{sec:experiments}
We describe the results of numerical experiments on synthetic and real data in this section.  We employ Algorithm~\ref{alg:poset_stability_discrete} with with both $\Psi = \Psi_{\text{stable}}$ and $\Psi = \Psi_{\text{test}}$.  For the testing-based approach, the manner in which p-values are obtained is described in the context of each application and we set $\alpha$ equal to $0.05 / |\mathcal{S}|$ for a given set $\mathcal{S}$ of minimal covering pairs.  For the stability-based approach, we consider $B=100$ subsamples obtained by partitioning a given dataset $50$ times into subsamples of equal size and we set $\alpha = 0.3$.

To obtain a desired level of expected false discovery with the stability-based approach, we appeal to Theorem~\ref{prop:refined_stability} as follows.  In the bound \eqref{eqn:bound_refined}, each $q_k$ can be derived by averaging over subsamples (as explained in the discussion after the statement of Theorem~\ref{prop:refined_stability}) and all the other quantities are known.  The values of these $q_k$'s in turn depend on the model estimates returned by the base procedure $\hatx_{\mathrm{base}}$ employed on the subsamples; in particular, if the estimate is the least element then each $q_k$ equals zero, and as $\hatx_{\mathrm{base}}$ returns models of increasing complexity, the value of each $q_k$ generally increases.  Building on this observation, we tune parameters in $\hatx_{\mathrm{base}}$ to return increasingly more complex models until the bound \eqref{eqn:bound_refined} is at the desired level.  For causal structure learning we employ Greedy Equivalence Search as our base procedure with tuning via the regularization parameter that controls model complexity \cite{Chickering2002OptimalSI}.  For clustering, we employ $k$-means \cite{Lloyd1982LeastSQ} as the base procedure with tuning via the number of clusters.  For our illustrations with ranking problems (both partial and total) in which we are provided with pairwise comparison data, our base procedure first employs the maximum-likelihood estimator associated to the Bradley-Terry model \cite{Bradley1952RankAO}, which returns a vector of positive weights $\hat{w}$ of dimension equal to the number of items.  Using this $\hat{w}$ we associate numerical values to covering pairs; each covering pair corresponds to increasing the complexity of a model by including a pair of items $(i,j)$ to the inversion set (in total ranking) or to the relation specifying a strict partial order (in partial ranking), and the value we assign is the difference $\hat{w}_j - \hat{w}_i$.  Our base procedure then constructs a path starting from the least element by greedily adding covering pairs of largest value at each step, provided these values are larger than a regularization parameter $\lambda > 0$; smaller values of $\lambda$ yield model estimates of larger complexity, while larger values yield estimates of smaller complexity.

\begin{figure*}[t]
\centering
\includegraphics[scale = 0.525]{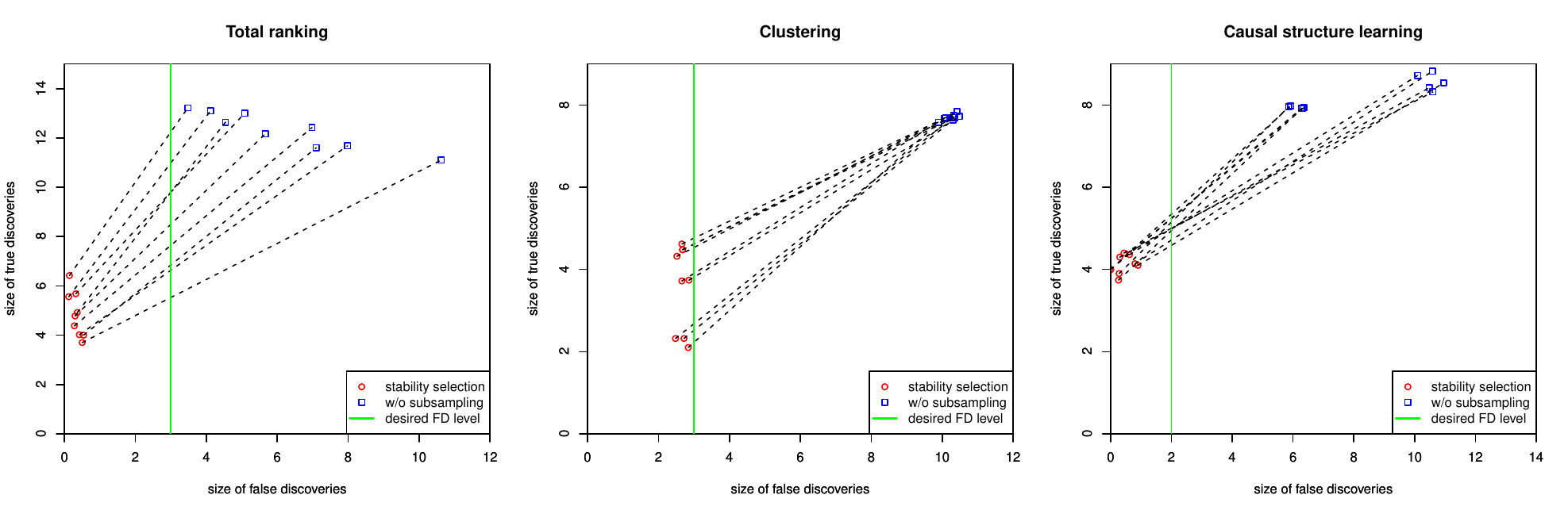}
\caption{Comparing the performance of Algorithm~\ref{alg:poset_stability_discrete} with $\Psi = \Psi_\text{stable}$ versus a non-subsampling approach for total ranking, clustering, and causal structure learning.  Each problem setting corresponds to a pair of dots and a connecting line. The comparison is in terms of the amount of false and true discoveries.}
\label{fig:synthetic}
\end{figure*}

Finally, for causal structure learning, we restrict our search during the model aggregation phase of Algorithm~\ref{alg:poset_stability_discrete} to paths that yield CPDAG models in which each connected component in the skeleton has diameter at most two; such a restriction facilitates a simple characterization of covering pairs.  This restriction is not imposed on the output of the base procedure. Moreover, the true model can be an arbitrary CPDAG.

\subsection{Synthetic data}
We describe experiments with synthetic data using Algorithm~\ref{alg:poset_stability_discrete} with $\Psi = \Psi_{\text{stable}}$.

\par{\textbf{Total ranking}}: We consider a total ranking problem with $p=30$ items.  We observe $n$ i.i.d. games between players $i,j$ with the outcome modeled as $y_{ij\ell} \sim \texttt{Bernoulli}(w^\star_i/(w^\star_i+w^\star_j))$ for $\ell = 1,\dots,n$, where $w^\star \in \R^p_{++}$ is a feature vector and $n \in \{200,250,300\}$.  We fix $w^\star$ by first defining $\tilde{w} \in \R^p_{++}$ as $\tilde{w}_i = \tau^{i-1}, ~ i=1,\dots,p$ for $\tau \in \{0.97,0.98,0.99\}$, and then setting $w^\star$ equal to a permutation of $\tilde{w}$ in which we swap the entries $1,3$, the entries $8,10$, the entries $15,17$, the entries $20,22$, and the entries $25,27$.  Smaller values of $\tau$ correspond to better-distinguished items, and hence to easier problem instances.  The base procedure is tuned such that the expected false discovery in \eqref{eqn:bound_refined} is at most three.

\par{\textbf{Clustering}}: We consider a clustering problem with $p=20$ variables.  The true partition consists of $12$ clusters with five variables in one cluster, another five variables in a second cluster, and the remaining variables in singleton clusters.  The $p$ variables are independent two-dimensional Gaussians.  Each variable in cluster $i$ has mean $(\mu_i,0)$ and covariance $\tfrac{1}{4}I$; each $\mu_i = i/d$ for $d \in \{3,3.5,4\}$.  Smaller values of $d$ correspond to better-separated clusters, and hence to easier problem instances.  We are provided $n$ i.i.d. observations of these variables for $n \in \{40,65,90\}$.  The base procedure is tuned such that the expected false discovery in \eqref{eqn:bound_refined} is at most three.

\par{\textbf{Causal structure learning}}: We consider a causal structure learning problem over $p=10$ variables.  The true DAG is generated by considering a random total ordering of the variables, drawing directed edges from higher nodes in the ordering to lower nodes independently with probability $v \in \{0.13,0.18\}$, and defining a linear structural causal model in which each variable is a linear combination of its parents plus independent Gaussian noise with mean zero and variance $\tfrac{1}{4}$.  The coefficients in the linear combination are drawn uniformly at random from the interval $[0.5,0.7]$.  Larger values of $v$ lead to denser DAGs, and hence to harder problem instances.  We obtain $n$ i.i.d. observations from these models for $n \in \{1000, 1200, 1400, 1600, 1800\}$.  The base procedure is tuned such that the expected false discovery in \eqref{eqn:bound_refined} is at most two.

For the preceding three problem classes, we compare the performance of our stability-based methodology versus that of a non-subsampled approach in which the base procedure (with suitable regularization) is applied to the entire dataset.  For total ranking, the non-subsampled procedure simply extracts the ranking implied by the maximum-likelihood estimator associated to the Bradley-Terry model.  For clustering, the non-subsampled approach employs $k$-means where the number of clusters is chosen to maximize the average silhoutte score \cite{Rousseeuw1987SilhouettesAG}.  For causal structure learning, the non-subsampled approach applies Greedy Equivalence Search with a regularization parameter chosen based on holdout validation ($70\%$ of the data is used for training and the remaining $30\%$ for validation).  Figure~\ref{fig:synthetic} presents the results of our experiments averaged over $50$ trials, and as the plots demonstrate, our stability-based methods yield models with smaller false discovery than the corresponding non-subsampled approaches.  This reduction in false discovery comes at the expense of a loss in power, which is especially significant for some of the harder problem settings.  However, in all cases our stability-based method provides the desired level of control on expected false discovery.

\subsection{Real data} We describe next experiments with real data.

\par{\textbf{Partial ranking of tennis players}}: We consider the task of partially ranking six professional tennis players -- Berdych, Djokovic, Federer, Murray, Nadal, and Wawrinka -- based on historical head-to-head matches of these players up to the end of 2022.  We apply Algorithm~\ref{alg:poset_stability_discrete} with $\Psi = \Psi_{\text{stable}}$ and with the base procedure tuned such that the expected false discovery in \eqref{eqn:bound_refined} is at most three.  The output of our procedure is a rank-nine model given by the partial ranking \{Djokovic, Nadal\} $>$ \{Berdych, Murray, Wawrinka\} and \{Federer\} $>$ \{Berdych, Wawrinka\}.

\par{\textbf{Total ranking of educational systems}}: We consider the task of totally ordering $p = 15$ OECD countries in reading comprehension based on test results from the Programme for International Student Assessment (PISA).  We take the null ranking as the ordering of the countries based on performance in $2015$ (see the first row in Table~\ref{table_pisa}), and we wish to update this model based on 2018 test scores (data obtained from \cite{KevinWang}), with the number of test scores ranging from $696$ to $3414$. We apply Algorithm~\ref{alg:poset_stability_discrete} with $\Psi = \Psi_{\text{test}}$ and we obtain p-values by modeling the average test score of each country as a Gaussian.  We set $\alpha = 0.05/\tfrac{p(p-1)}{2}$ (here $\tfrac{p(p-1)}{2}$ is the cardinality of a set of minimal covering pairs), which yields the guarantee from Theorem~\ref{thm:testing_general} that the estimated model has zero false discovery with probability at least $0.95$.  The output of our procedure is the rank-nine model given by the total ranking in the second row in Table~\ref{table_pisa}.

\begin{table}[h!]
\scalebox{0.76}{
    \begin{tabular}{|c|c|c|c|c|c|c|c|c|c|c|c|c|c|c|c|}
    \hline
    & 1&2&3&4&5&6&7&8&9&10&11&12&13&14&15\\
    \hline
    2015 base ranking & CAN&FIN&IRL&EST&KOR&JPN&NOR&NZL&DEU&POL&SvN&NLD&AUS&SWE&DNK\\\hline
   testing approach & FIN&IRL&EST&CAN&KOR&JPN&NOR&NZL&POL&DEU&AUS&SWE&SvN&DNK&NLD\\
    \hline
\end{tabular}}
\caption{Ranking of nations according to PISA reading comprehension scores; the first column is the 2015 ranking of $15$ OECD countries which serves as the base ranking for our analysis: based on test results in 2018, we update this ranking using Algorithm~\ref{alg:poset_stability_discrete} based on $\Psi = \Psi_\text{test}$ with the result shown in the second column.}
\label{table_pisa}
\end{table}

\par{\textbf{Learning causal structure among proteins}}: We aim to learn causal relations underlying $p = 11$ phosphoproteins and phospholipids from a mass spectroscopy dataset containing $854$ measurements of abundance levels in an observational setting \cite{sachs}.  We apply Algorithm~\ref{alg:poset_stability_discrete} with $\Psi = \Psi_\mathrm{stable}$ and with the base procedure tuned such that the expected false discovery in \eqref{eqn:bound_refined} is at most two.  Figure~\ref{fig:sachs} presents the rank-six CPDAG model obtained from our algorithm and compares to the estimates obtained from the literature \cite{nicolai_pnas,sachs,IGSP}.  Our CPDAG estimate has fewer edges than those in \cite{nicolai_pnas,sachs,IGSP}, which do not explicitly provide control on false discovery.  

\begin{figure}[!htbp]
\begin{minipage}{0.45\linewidth}
\usetikzlibrary{graphs}
\begin{tikzpicture}
    \tikzstyle{every node}=[font=\footnotesize]
    \graph[clockwise, radius=1.2cm, n=11]
    {
       p38/p38, ERK/ERK, MEK/MEK, JNK/JNK, PIP3/PIP3, PKA/PKA, PKC/PKC, PIP2/PIP2, PLcG/PLcG,AKT/AKT,RAF/RAF;
    };
    \draw[-] (RAF)--(MEK);
    \draw[-] (PIP2)--(PIP3);
    \draw[-] (AKT)--(ERK);
    \draw[-] (AKT)--(PKA);
     \draw[->] (p38)--(PKC);
     \draw[->] (JNK)--(PKC);
\end{tikzpicture}
\end{minipage}
\hspace{-0.2in}
\begin{minipage}{0.45\linewidth}
\scalebox{0.8}{
\begin{tabular}{cccccccc}\hline
      Edge & \cite{sachs}a & \cite{sachs}b & \cite{mooij} & \cite{Eaton} & \cite{nicolai_pnas}a & \cite{nicolai_pnas}b & \cite{IGSP}  \\ \hline
       MEK -- RAF & $\leftarrow$ &$\leftarrow$ & $\rightarrow$ & $\rightarrow$ & & -- & -- \\
       AKT -- ERK &  & $\leftarrow$&$\rightarrow$ & $\leftarrow$ &-- & -- & -- \\
       PIP3 -- PIP2  &$\rightarrow$ & $\rightarrow$& $\rightarrow$ &$\leftarrow$&$\rightarrow$&$\rightarrow$&$\rightarrow$ \\ 
       AKT -- PKA  &  $\leftarrow$&$\leftarrow$ & $\leftarrow$ & $\leftarrow$& &$\leftarrow$&$\rightarrow$ \\
       P38 $\rightarrow$ PKC & $\leftarrow$ & $\leftarrow$ & $\leftarrow$ & $\leftarrow$ & & -- & \\
       JNK $\rightarrow$ PKC & $\leftarrow$ & $\leftarrow$ & $\leftarrow$ & $\leftarrow$ & $\leftarrow$& -- & \\
        \hline
      \end{tabular}}
\end{minipage}
\caption{\textbf{left}: CPDAG obtained by Algorithm~\ref{alg:poset_stability_discrete} with $\Psi = \Psi_\text{stable}$; \textbf{right}: comparing the edges obtained by our algorithm (shown in the leftmost column) with different causal discovery methods (with indicated reference).  The consensus network according to \cite{sachs} is denoted here by ``\cite{sachs}a'' and their reconstructed network by ``\cite{sachs}b''; The authors in \cite{nicolai_pnas} apply two methods, and the results are presented by ``\cite{nicolai_pnas}a'' and ``\cite{nicolai_pnas}b''. Here, ``$-$'' means that the edge direction is not identified.}
\label{fig:sachs}
\end{figure}
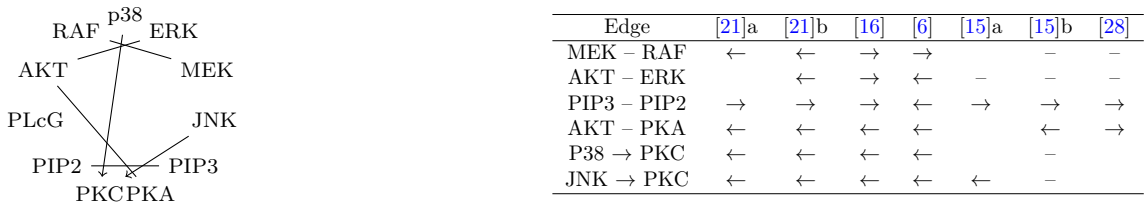

\section{Proofs}
\label{sec:proofs}
For notational ease, for a covering pair $(u,v)$ and element $z$ in the poset $\Lp$, we define $f(u,v;z) \triangleq \rho(v,z)-\rho(u,z)$. Recall that $\mathcal{T}_\text{null} \triangleq \{(u,v) \text{ covering pair in } \Lp ~|~ \rho(v,x^\star)=\rho(u,x^\star)\}$. Our analysis relies on the following lemmas with the proofs presented in Appendix~\ref{sec:proof_of_main_paper_lemmas}.
\begin{lemma}Fix a discrete model poset $\Lp$ with integer-valued similarity valuation $\rho$.  For any model $x \in \Lp$ with $(x_0,\dots,x_k)$ being any path from the least element $x_0 = x_{\text{least}}$ to $x_k = x$, we have that $\FD(x,x^\star) \leq \sum_{i=1}^k \mathbb{I}[(x_{i-1},x_i) \in \mathcal{T}_{\text{null}}]$.  As a result, we have that $\FD(x,x^\star) > 0$ implies the existence of some $i$ for which $(x_{i-1},x_i) \in \mathcal{T}_{\text{null}}$.
\label{lemma:discrete_integer}
\end{lemma}

\begin{lemma} For any covering pairs $(u,v)$ and $(x,y)$ with $v \preceq x$, we cannot have that $f(u,v; z) = f(x,y; z)$ for all $z \in \Lp$.
\label{lemma:two_covering_pairs}
\end{lemma}

\subsection{Proof of Theorem~\ref{thm:main_stability}}
\label{proof_stability}
For notational convenience, we let $\hat{x}^{(\ell)}_\mathrm{base} =  \hat{x}_\mathrm{base}(\mathcal{D}^{(\ell)})$ where $\{\mathcal{D}^{(\ell)}\}_{\ell=1}^B$ are the subsamples of $\mathcal{D}$.  Let $\hat{x}_{\text{stable}}$ be the output of Algorithm~\ref{alg:poset_stability_discrete} with $\rk(\hat{x}_{\text{stable}}) = \hat{k}$, and let $(x_0,\dots,x_{\hat{k}})$ be the associated path from the least element $x_0 = x_{\text{least}}$ to $x_{\hat{k}} = \hat{x}_{\text{stable}}$; we have that $\frac{1}{B}\sum_{\ell=1}^B{f(x_{i-1},x_i;\hat{x}_\mathrm{base}^{(\ell)})}/{c_\Lp(x_{i-1},x_i)} \geq(1-\alpha)$ for each $i = 1,\dots,\hat{k}$.  Let $\mathcal{C} \triangleq \{(x_{i-1},x_i) ~|~ i=1,\dots,\hat{k}\}$.  From Lemma~\ref{lemma:discrete_integer}, we also have that $\FD(\hat{x}_\text{stable},x^\star) \leq \left|\mathcal{C} \cap \mathcal{T}_{\text{null}} \right|$.  Combining these observations, we conclude that $\FD(\hat{x}_\text{stable},x^\star) \leq \sum_{(u,v) \in \mathcal{C} \cap \mathcal{T}_{\text{null}}} \mathbb{I}\left[\frac{1}{B}\sum_{\ell=1}^B \frac{f(u,v;\hat{x}_\mathrm{base}^{(\ell)})}{c_\Lp(u,v)} \geq 1-\alpha\right]$.  Next, we observe that for each covering pair in $\mathcal{C}$ there exists a covering pair in the minimal set $\mathcal{S}$ with the values of $f$ and $c_{\Lp}$ remaining the same; moreover, distinct covering pairs in $\mathcal{C}$ map to distinct covering pairs in $\mathcal{S}$ from Lemma~\ref{lemma:two_covering_pairs}.  Thus, we conclude that $\FD(\hat{x}_\text{stable},x^\star) \leq \sum_{{(u,v)\in\mathcal{S}\cap \mathcal{T}_\text{null}}} \mathbb{I}\left[\frac{1}{B}\sum_{\ell=1}^B \frac{f(u,v;\hat{x}_\mathrm{base}^{(\ell)})}{c_\Lp(u,v)}\geq 1-\alpha\right]$.  We then have the following sequence of steps:

\begin{equation}
\begin{aligned}
   \FD(\hat{x}_\text{stable},x^\star)  &\leq \sum_{\substack{(u,v)\in\\\mathcal{S}\cap \mathcal{T}_\text{null}}} \mathbb{I}\left[\frac{1}{B/2}\sum_{\ell=1}^{B/2}\sum_{i\in\{0,1\}} \hspace{-0.05in}\frac{f(u,v;\hat{x}_\mathrm{base}^{(2\ell-i)})}{c_\Lp(u,v)} \geq 2-2\alpha\right]\\
    &\leq  \sum_{\substack{(u,v)\in\\\mathcal{S}\cap \mathcal{T}_\text{null}}} \mathbb{I}\left[\frac{1}{B/2}\sum_{\ell=1}^{B/2}\prod_{i\in\{0,1\}} \hspace{-0.07in}\frac{f(u,v;\hat{x}_\mathrm{base}^{(2\ell-i)})}{c_\Lp(u,v)} \geq 1-2\alpha\right].
    \end{aligned}
    \label{eq:1}
\end{equation}
The second inequality follows from $ab \geq a+b-1$ for $a,b \in [0,1]$, where we set $a =  f(u,v;\hat{x}_\mathrm{base}^{(2\ell-1)})/{c_\Lp(u,v)}$ and  $b = {f(u,v;\hat{x}_\mathrm{base}^{(2\ell)})}/{c_\Lp(u,v)}$, and note that $f(u,v;z)/c_\Lp(u,v) \in [0,1]$ for any $z \in \Lp$.  Taking expectations on both sides of the preceding inequality, we finally seek a bound on $\mathbb{P}\left[\frac{1}{B/2}\sum_{\ell=1}^{B/2}\prod_{i\in\{0,1\}} \frac{f(u,v;\hat{x}_\mathrm{base}^{(2\ell-i)})}{c_\Lp(u,v)} \geq 1-2\alpha\right]$.  We have that:
\begin{equation}
\begin{aligned}
    \mathbb{P}\left[\frac{1}{B/2}\sum_{\ell=1}^{B/2}\prod_{i\in\{0,1\}} \frac{f(u,v;\hat{x}_\mathrm{base}^{(2\ell-i)})}{c_\Lp(u,v)} \geq 1-2\alpha\right]
    &\leq \frac{\mathbb{E}\left[\frac{1}{B/2}\sum_{\ell=1}^{B/2}\prod_{i\in\{0,1\}} \frac{f(u,v;\hat{x}_\mathrm{base}^{(2\ell-i)})}{c_\Lp(u,v)}\right]}{1-2\alpha}\\&= \frac{\mathbb{E}\left[f(u,v;\hatx_\mathrm{sub})\right]^2}{c_\Lp(u,v)^2(1-2\alpha)}.
    \end{aligned}
        \label{eq:2}
\end{equation}
Here $\hatx_\mathrm{sub}$ represents the estimator corresponding to the base procedure $\hat{x}_\mathrm{base}$ applied to a subsample of $\mathcal{D}$ of size $|\mathcal{D}|/2$.  The inequality follows from Markov's inequality, and the equality follows by noting that complementary bags are independent and identically distributed. Combining \eqref{eq:1} and \eqref{eq:2}, we obtain the desired result.

\subsection{Proof of Theorem~\ref{prop:refined_stability}}
\label{sec:proof_stability_refined}
\label{sec:proof_after_assumptions_stability}
We have from Theorem~\ref{thm:main_stability} that:
\begin{equation*}
\E[\FD(\hat{x}_\text{stable},x^\star)]\leq \sum_{k=1}^{\rk(\Lp)}\sum_{(u,v)\in\setS_k \cap \mathcal{T}_\text{null}} \frac{\mathbb{E}[f(u,v;\hatx_\mathrm{sub})]^2}{(1-2\alpha)c_\Lp(u,v)^2}.
\end{equation*}
Our goal is to bound $\E[f(u,v;\hat{x}_\mathrm{sub})]/c_\Lp(u,v)$ for $(u,v)\in \mathcal{S}_k \cap \mathcal{T}_\text{null}$. Note that each $q_k$ may be decomposed as
\begin{equation*}
\begin{aligned}
    q_k &= \sum_{\substack{(u,v)\in \\\mathcal{S}_k \cap \mathcal{T}_\text{null}}} \frac{\E[f(u,v;\hat{x}_\mathrm{sub})]}{c_\Lp(u,v)} +  \sum_{\substack{(u,v)\in \\\mathcal{S}_k \setminus \mathcal{T}_\text{null}}} \frac{\E[f(u,v;\hat{x}_\mathrm{sub})]}{c_\Lp(u,v)}. 
\end{aligned}
\end{equation*}
Appealing to Assumption 1, we have that 
\begin{equation*}
\begin{aligned}
    q_k &\geq \left(1+\frac{|\setS_k \setminus \mathcal{T}_\text{null}|}{|\setS_k \cap \mathcal{T}_\text{null}|}\right)\sum_{(u,v)\in \mathcal{S}_k \cap \mathcal{T}_\text{null}} \frac{\E[f(u,v;\hat{x}_\mathrm{sub})]}{c_\Lp(u,v)}. 
\end{aligned}
\end{equation*}
Rearranging the terms, we obtain that
\begin{equation*}
   \sum_{(u,v)\in \mathcal{S}_k \cap \mathcal{T}_\text{null}} \frac{\E[f(u,v;\hat{x}_\mathrm{sub})]}{c_\Lp(u,v)}  \leq \frac{q_k}{|\mathcal{S}_k|} |\mathcal{S}_k \cap \mathcal{T}_\text{null}|.
\end{equation*}
Appealing to Assumption 2, we have for each $(u,v)\in \mathcal{S}_k \cap \mathcal{T}_\text{null}$ that $\frac{\E[f(u,v;\hat{x}_\mathrm{sub})]}{c_{\Lp}(u,v)} \leq \frac{q_k}{|\mathcal{S}_k|}$. Plugging this bound into the conclusion of Theorem~\ref{thm:main_stability} yields the desired result.

\subsection{Proof of Theorem~\ref{thm:testing_general}}
\label{proof_thm_testing_general}
Let $\hat{x}_{\text{test}}$ be the output of Algorithm~\ref{alg:poset_stability_discrete} with $\rk(\hat{x}_{\text{test}}) = \hat{k}$, and let $(x_0,\dots,x_{\hat{k}})$ be the associated path from the least element $x_0 = x_{\text{least}}$ to $x_{\hat{k}} = \hat{x}_{\text{test}}$; we have that $\Psi_{\text{test}}(x_{i-1},x_i) \leq \alpha$ for each $i = 1,\dots,\hat{k}$.  Let $\mathcal{C} \triangleq \{(x_{i-1},x_i) ~|~ i=1,\dots,\hat{k}\}$.  From Lemma~\ref{lemma:discrete_integer}, we have that $\FD(\hat{x}_{\text{test}}, x^\star) > 0$ implies the existence of a covering pair $(u,v) \in \mathcal{C} \cap \mathcal{T}_{\text{null}}$ for which $\Psi_{\text{test}}(u,v) \leq \alpha$.  For each covering pair in $\mathcal{C}$, there exists a covering pair in $\mathcal{S}$ with the same value of $\Psi_{\text{test}}$; thus, there exists $(u,v) \in \mathcal{S} \cap \mathcal{T}_{\text{null}}$ such that $\Psi_{\text{test}}(u,v) \leq \alpha$.  Consequently:
\begin{equation}
\begin{aligned}
\mathbb{P}(\FD(\hat{x}_{\text{test}}, x^\star) > 0) &\leq \mathbb{P}\left(\exists (u,v) \in \setS\cap\mathcal{T}_\text{null}  \text{ s.t. } \Psi_{\text{test}}(u,v) \leq \alpha  \right) \\ &\leq \sum_{(u,v) \in \setS\cap\mathcal{T}_\text{null}} \mathbb{P}(\Psi_{\text{test}}(u,v) \leq \alpha) \leq \alpha |\mathcal{S}|. 
\end{aligned}
\label{eqn:test}
\end{equation}
Here the second inequality follows from the union bound and the final inequality follows from the fact that the random variable $\Psi_\text{test}(u,v)$ is a valid p-value under the null hypothesis $\rho(v,x^\star)=\rho(u,x^\star)$.

\subsection{Proof of Proposition~\ref{prop:join}}
\label{proof:prop_join}
For each $\hat{x}^{(j)}, ~ j=1,\dots,m$, we are given that there is a path from $x_{\text{least}}$ to $\hat{x}^{(j)}$ such that $\Psi_{\text{test}}$ is bounded by $\alpha$ for each covering pair in the path; let $\mathcal{C}^{(j)}$ be the set of these covering pairs.  As described in Section~\ref{sec:testing} in the discussion preceding Proposition~\ref{prop:join}, $
\FD(\hat{x}_\text{join},x^\star) > 0$ implies that $\FD(\hat{x}^{(j)},x^\star) > 0$ for some $j = 1,\dots,m$, which in turn implies from Lemma~\ref{lemma:discrete_integer} the existence of a covering pair $(u,v) \in \mathcal{C}^{(j)} \cap \mathcal{T}_{\text{null}}$ for some $j =1,\dots,m$. Following the same logic as in the proof of Theorem~\ref{thm:testing_general}, we conclude that $
\FD(\hat{x}_\text{join},x^\star) > 0$ implies the existence of $(u,v) \in \mathcal{S} \cap \mathcal{T}_{\text{null}}$ such that $\Psi_{\text{test}}(u,v) \leq \alpha$.  Using the same reasoning as in \eqref{eqn:test}, we have the desired conclusion.

\section{Discussion} \label{sec:discussion} We present a general framework to endow a collection of models with poset structure.  This framework yields a systematic approach for quantifying model complexity and false positive error in an array of complex model selection tasks in which models are not characterized by Boolean logical structure (such as in variable selection).  Moreover, we develop methodology for controlling false positive error in general model selection problems over posets, and we describe experimental results that demonstrate the utility of our framework.

We finally discuss some future research questions that arise from our work.  On the mathematical front, a basic open question is to characterize fundamental tradeoffs between false positive and false negative errors that are achievable by any procedure in model selection over a general poset; this would generalize the Neyman-Pearson lemma on optimal procedures for testing between two hypotheses.  On the computational and methodological front, it is of interest to develop new methods to control false positive error as well as false discovery rates, including in settings involving continuous model posets.

\section*{Acknowledgements}
We thank Marina Meila and Lior Pachter for insightful conversations. AT received funding from the Royalty Research Fund at the University of Washington.
PB received funding from the European Research Council under the European Union’s Horizon 2020 research and innovation program (grant agreement No. 786461).  VC was supported in part by the Air Force Office of Scientific Research grants FA9550-22-1-0225 and FA9550-23-1-0204 and by the National Science Foundation grant DMS 2113724.

\bibliography{main}
\bibliographystyle{abbrv}

\newpage
\appendix
\section*{\centering Appendix}

\section{Meet Semi-lattice and Join Semi-lattice Properties and Posets in Examples~\ref*{ex:variable-selection}-\ref*{ex:blind-source-separation}}
\label{sec:meet_semi}
The Boolean poset (Example~\ref*{ex:variable-selection}), partition poset (Examples~\ref*{ex:clustering}-\ref*{ex:multisample-testing}), integer poset (Example~\ref*{ex:multiple-changepoint}), permutation poset (Example~\ref*{ex:complete-ranking}), and subspace poset (Example~\ref*{ex:subspace-estimation}) are all known in the literature to be lattices (and consequently meet-semi and join semi-lattices); see \cite{stanley_2011}. 

We next show that for Examples~\ref*{ex:partial-ranking} and \ref*{ex:blind-source-separation} associated with partial ranking and blind-source separation, the corresponding posets are also meet semi-lattices. Consider the partial ranking setting in Example~\ref*{ex:partial-ranking}. Let $\mathcal{R}_1$ and $\mathcal{R}_2$ be two relations that are irreflexive, asymmetric, and transitive. Recalling that the partial ordering is based on inclusion, it is clear that the relations $\mathcal{R} = \{(a,b): (a,b) \in \mathcal{R}_1, (a,b) \in \mathcal{R}_2\}$ is the unique largest rank element in the partial ranking poset such that $\mathcal{R} \preceq \mathcal{R}_1$ and $\mathcal{R} \preceq \mathcal{R}_2$. Furthermore, for any $\tilde{\mathcal{R}}$ with  $\tilde{\mathcal{R}} \preceq \mathcal{R}_1$ and $\tilde{\mathcal{R}} \preceq \mathcal{R}_2$, we clearly have that $\tilde{\mathcal{R}} \preceq \mathcal{R}$. Consider the blind-source separation setting in Example~\ref*{ex:blind-source-separation}. Let $x_1$ and $x_2$ be two sets of linearly independent subsets of unit norm vectors. Recalling that the partial ordering in the associated poset is based on inclusion, it is clear that the set $y = x_1 \cap x_2$ is the unique largest rank element in the partial ranking poset such that $y \preceq x_1$ and $y \preceq x_2$. Furthermore, for every $z$ with $z \preceq x_1$ and $z \preceq x_2$, we have that $z \preceq y$.

We show that the poset corresponding to causal structure learning setting (Example~\ref*{ex:causal-learning}) is not meet semi-lattice or join semi-lattice. As a counterexample, consider the CPDAGs $\mathcal{C}_i$ for $i = 1,2,3,4$ shown in Figure~\ref{fig:CPDAG_ex}. Notice that $\mathcal{C}_3 \preceq \mathcal{C}_1$, $\mathcal{C}_3 \preceq \mathcal{C}_2$, $\mathcal{C}_4 \preceq \mathcal{C}_1$, and $\mathcal{C}_4 \preceq \mathcal{C}_2$. Notice also that $\mathcal{C}_3$ and $\mathcal{C}_4$ are both CPDAGs with the largest rank that are smaller (in a partial order sense) than $\mathcal{C}_1$ and $\mathcal{C}_2$. We thus can conclude that the poset is not meet semi-lattice. Similarly, $\mathcal{C}_1$ and $\mathcal{C}_2$ are both CPDAGs with the smallest rank that are larger (in a partial order sense) than $\mathcal{C}_3$ and $\mathcal{C}_4$. We thus can conclude that the poset is not join semi-lattice.

We next show that the poset for Example~\ref*{ex:partial-ranking} is not join semi-lattice with a simple counterexample. Consider as an example elements $x_1 = \{(1,2)\}$ and $x_2 = \{(2,1)\}$. Note that there does not exist an element $z$ such that $x_1 \preceq z$ and $x_2 \preceq z$. Thus, the poset is not join semi-lattice.

Finally, we show that the poset corresponding to blind-source separation (Example~\ref*{ex:blind-source-separation}) is not join semi-lattice. Consider a collection of $p+1$ rank-1 elements in this poset, each element consisting of a single $p$ dimensional vector. Then, evidently, there cannot exist an element $z$ consisting of a set of vectors that contains all of the vectors in the rank-1 elements, while satisfying the linear independence condition.  

\begin{figure}[h!]
\centering
\subfloat[$\mathcal{C}_1$]{
\begin{tikzpicture}[main/.style = {draw, circle}] 
\node[main] (1) {$1$}; 
\node[main] (2) [below left of=1]  {$2$}; 
\node[main] (3) [below right of=1] {$3$}; 
\draw[->] (2) -- (1);
\draw[->] (3) -- (1);
\end{tikzpicture}} 
\hspace{0.2in}
\subfloat[$\mathcal{C}_2$]{\begin{tikzpicture}[main/.style = {draw, circle}] 
\node[main] (1) {$1$}; 
\node[main] (2) [below left of=1]  {$2$}; 
\node[main] (3) [below right of=1] {$3$}; 
\draw (2) -- (1);
\draw (3) -- (1);
\end{tikzpicture}} 
\hspace{0.2in}
\subfloat[$\mathcal{C}_3$]{\begin{tikzpicture}[main/.style = {draw, circle}] 
\node[main] (1) {$1$}; 
\node[main] (2) [below left of=1]  {$2$}; 
\node[main] (3) [below right of=1] {$3$}; 
\draw (2) -- (1);
\end{tikzpicture}} 
\hspace{0.2in}
\subfloat[$\mathcal{C}_4$]{\begin{tikzpicture}[main/.style = {draw, circle}] 
\node[main] (1) {$1$}; 
\node[main] (2) [below left of=1]  {$2$}; 
\node[main] (3) [below right of=1] {$3$}; 
\draw (3) -- (1);
\end{tikzpicture}}
\caption{Four CPDAGs. Here, CPDAGs $\mathcal{C}_3$ and $\mathcal{C}_4$ are both largest complexity models that are smaller (in partial order sense) than $\mathcal{C}_1$ and $\mathcal{C}_2$. Similarly, CPDAGs $\mathcal{C}_1$ and $\mathcal{C}_2$ are the smallest complexity models that are larger (in a partial order sense) than $\mathcal{C}_3$ and $\mathcal{C}_4$.}
\label{fig:CPDAG_ex}
\end{figure}
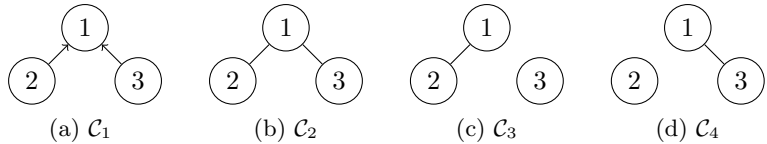

\section{{Gradedness of the CPDAG Poset}}
\label{sec:graded_cpdag}
To show gradedness of the CPDAG poset, we rely on the following classical result in \cite{Chickering2002OptimalSI}. 
\begin{proposition}[Theorem 4 of \cite{Chickering2002OptimalSI}]The conditional dependencies of the DAG $\cg_1$ are contained in the conditional dependencies of the DAG $\cg_2$ if and only if there is a sequence of allowed edge reversals $R$ (these are edge reversals that do not create or remove v-structures so do not change the conditional dependency relationships) and edge deletions $D$ that once applied to $\cg_2$ yield $\cg_1$.
\label{prop:Chickering}
\end{proposition}
We are now ready to prove that the CPDAG poset is graded. As described in Example \ref{ex:causal-learning}, we use the rank function that computes the number of edges in the input CPDAG. The first property of gradedness is that if $\cc_1 \preceq \cc_2$, then $\mathrm{rank}(\cc_1) \leq \mathrm{rank}(\cc_2)$. To prove this, suppose an edge exists between a pair of nodes $(i,j)$ in $\cc_1$ that are not connected in $\cc_2$. The presence of this edge in $\cc_1$ implies that in the model $\cc_1$, the random variables corresponding to $(i,j)$ are conditionally dependent, conditioned on any other collection of nodes in the model $\cc_1$. However, the absence of the edge in $\cc_2$ implies that in the model $\cc_2$, the random variables corresponding to nodes $i,j$ are conditionally independent, conditioned on some collection of nodes in the model $\cc_1$. This leads to a contradiction as otherwise not every conditional dependency encoded by $\cc_1$ would be encoded by $\cc_2$. Thus, every pair of nodes connected in $\cc_1$ must be connected in $\cc_2$, allowing us to conclude that $\mathrm{rank}(\cc_1) \leq \mathrm{rank}(\cc_2)$. The second property of gradedness is that if the CPDAGs $(\cc_1,\cc_2)$ are covering pairs, then $\mathrm{rank}(\cc_1)+1 = \mathrm{rank}(\cc_2)$. Suppose as a point of contradiction that $\cc_2$ has two or more edges than $\cc_1$. Our goal is to show that we can identify a CPDAG $\cc$ in between $\cc_2$ and $\cc_1$.  To that end, take $\cg_2$ and $\cg_1$ as two DAGs in these CPDAGs. From Proposition~\ref{prop:Chickering}, there exists a sequence of edge reversals and edge deletions applied to $\cg_1$ to obtain $\cg_1$, where importantly, there must be at least two edge deletions. Take $\cg$ to be the DAG in the sequence only after one edge deletion. By appealing to Porposition~\ref{prop:Chickering}, all the conditional dependencies of $\cg$ are contained in $\cg_2$. Similarly, since there is a sequence of allowed edge reversals and deletions that when applied to $\cg$ yield $\cg_1$, all the conditional dependencies of $\cg_1$ are contained in $\cg$. Let $\cc$ be the CPDAG obtained by completing $\cg$. We have shown that $\cc \neq \cc_1,\cc_2$ and $\cc_1 \preceq \cc \preceq \cc_2$. This leads to a contradiction since $(\cc_1,\cc_2)$ are covering pairs.

\section{Proof that $\mathrm{Eq.}$ (1) is a Similarity Valuation Function}
\label{sec:meet_rho}
Recall that 
\begin{equation}\rho_\text{meet}(x,y) = \max_{z \preceq x, z \preceq y} \rk(z).
\label{eqn:rho_meet_1}
\end{equation}
By definition, $\rho_\text{meet}(\cdot,\cdot)$ is a symmetric function. We will now show that it satisfies the three properties in Definition~\ref*{defn:rho_prop} for any pair of elements $x,y \in \Lp$. For the first property, we can conclude $\rho_\text{meet}(x,y) \geq 0$ since by definition, the rank function returns a non-negative integer for all the elements in the poset. Again, because of the property of the rank function in a graded poset, a feasible $z$ (satisfying the constraints $z\preceq x$, $z \preceq y$) will necessarily have $\rk(z) \leq \min\{\rk(x),\rk(y)\}$. For the second property, consider any $w \in \Lp$ with $x \preceq w$. Note that:
\begin{equation}
\rho_\text{meet}(w,y) = \max_{z \preceq w, z \preceq y} \rk(z).
\label{eqn:rho_meet_2}
\end{equation}
Then, any feasible $z$ in \eqref{eqn:rho_meet_1} is also feasible in \eqref{eqn:rho_meet_2} by the transitive property of posets. Therefore, $\rho_\text{meet}(x,y) \leq \rho_\text{meet}(w,y)$.  For the third property, first note that if $x \preceq y$, then $z = x$ is feasible in \eqref{eqn:rho_meet_1} and thus $\rho_\text{meet}(x,y) \geq \rk(x)$. Since also $\rho_\text{meet}(x,y) \leq \rk(x)$ by the second property of similarity valuations, we have that $\rho_\text{meet}(x,y) = \rk(x)$. Now suppose that $\rho_\text{meet}(x,y) = \rk(x)$. By \eqref{eqn:rho_meet_1}, we conclude that there exists a feasible $z$ ($z\preceq x$, $z\preceq y$) such that $\rk(z) = \rk(x)$. By the property of the rank function, we have that if $\rk(z) = \rk(x)$ and $z \preceq x$, then $z = x$. Since we have additionally that $z \preceq y$, we conclude that $x \preceq y$. 

 \section{Specializing Bound $\mathrm{Eq.}$ (8) for Different Problem Settings}
\label{sec:setS}

\subsection{Partial Ranking}
Let $S = \{a_1,a_2,\dots,a_p\}$ be the set of $p$ elements. We use the similarity valuation $\rho:= \rho_\text{meet}$ in $\mathrm{Eq.}$ (1) of the main paper. 
\subsubsection{Characterizing $\setS$ for Partial Ranking}
We construct a set $\setS$ satisfying the properties in Definition~\ref*{defn:set_S} of the main paper. Specifically, we let:
$$\mathcal{S} = \{(a_i,a_j): i \neq j\},$$
with $|\setS_1| = p(p-1)$ and $\setS_k = \emptyset$ for every $k \geq 2$. 

We will show that set $\setS$ as constructed above satisfies Definition~\ref*{defn:set_S}. First, consider any covering pair $({u'},{v'}) \notin \setS$. Here, ${u'}$ and ${v'}$ are relations and ${v'} = u' \cup (a_i,a_j)$ for some $i \neq j$. Then, for any $z \in \Lp$, it is easy to see that
$$\rho({v'},z)-\rho({u'},z) = \mathbb{I}[(a_i,a_j) \in z] = \rho(v,z)-\rho(u,z),$$ 
where $v = \{(a_i,a_j)\}$ and $u= \emptyset$. Clearly, $\rk(v) \leq \rk({v'})$.

To show the second property, consider covering pairs $(\{(a_i,a_j)\},\emptyset) \in \setS$ and $(\{(a_k,a_l)\},\emptyset) \in \setS$. By construction of the set $\setS$, $(a_i,a_j) \neq (a_k,a_l)$. Let $z = \{(a_i,a_j)\}$. Then, it is straightforward to see that $\rho(\{(a_i,a_j)\},z) - \rho(\emptyset,z) = 1$ but $\rho(\{(a_k,a_l)\},z) - \rho(\emptyset,z) = 0$.

\subsubsection{Characterizing $c_\Lp(x,y)$ for Covering Pair $(x,y)$} Since for any $z$, $\rho(y,z)-\rho(x,z) = \mathbb{I}((a_i,a_j) \in z)$ for some $(a_i,a_j)$. Thus, $c_\Lp(x,y) = 1$.

\subsubsection{Refined False Discovery Bound for Partial Ranking}
Let $\hat{x}_\text{stable}$ be output of Algorithm~\ref*{alg:poset_stability_discrete} with $\Psi = \Psi_\text{stable}$. Then: 
$$\mathbb{E}[\FD(\hat{x}_\text{stable},x^\star)] \leq \frac{q_1^2}{(1-2\alpha)p(p-1)},$$
where 
$$q_1 = \sum_{i\neq j} \mathbb{I}[(a_i,a_j) \in \hat{x}_\text{sub}].$$
Here, $\hat{x}_\text{sub}$ is the estimated partial ranking from supplying $n/2$ samples to the base estimator. We can use the following data-driven approximation for $q_1$:$q_1 \approx \frac{1}{B}\sum_{\ell=1}^B\sum_{i \neq j} \mathbb{I}[(a_i,a_j) \in \hat{x}_\text{base}(\mathcal{D}^{(\ell)})]$ with $\hat{x}_\text{base}(\mathcal{D}^{(\ell)}), l = 1,2,\dots,B$ representing the estimates from subsampling. 

\subsection{Total Ranking}
\label{sec:total_ranking_setS}
Let $S = \{a_1,a_2,\dots,a_p\}$ be the set of $p$ elements. Let $\pi_\text{null}(a_i) = i$ for every $i = 1,2,\dots,p$. We use the similarity valuation $\rho:= \rho_\text{total-ranking}$ in $\mathrm{Eq.}$ (2) of the main paper. As each element in the poset corresponds to a function $\pi : S \to S$, we will use this functional notation throughout.  

\subsubsection{Characterizing $\setS$ for Total Ranking}
We construct a set $\setS$ satisfying the properties in Definition~\ref*{defn:set_S} of the main paper. Initialize $\setS = \emptyset$. Then, for every relation $(a_i,a_j)$ with $i < j$, we augment $\setS$ as follows:
$$\setS = \setS \cup (\pi_1,\pi_2),$$
where $\pi_1,\pi_2$ are covering pairs. Here, $\pi_2$ is any rank $j-i$ element in the poset with the relation $(a_i,a_j)$ in its corresponding inversion set. Furthermore, we let $\pi_1$ be a rank $j-i-1$ element that is covered by $\pi_2$ and does not contain $(a_i,a_j)$ in its inversion set. Recalling that $\setS_k = \{(\pi_1,\pi_2) \in \setS, \rk(\pi_2) = k\}$, we have that for every $k = 1,2,\dots,p-1$
$$|\setS_k| = p-k.$$

We will show that set $\setS$ as constructed above satisfies Definition~\ref*{defn:set_S}. First, consider any covering pair $(\tilde{\pi}_1,\tilde{\pi}_2) \notin \setS$. Then by definition, the corresponding inversion sets are nested, i.e. $\text{inv}(\tilde{\pi}_2;\pi_\text{null}) \supseteq \text{inv}(\tilde{\pi}_1;\pi_\text{null})$ with the difference being a single relation. We will denote this relation by $(a_i,a_j)$ with $j > i$. Consider the covering pair $(\pi_1,\pi_2) \in \setS$ where $(a_i,a_j)$ is in the inversion set of $\pi_2$ but not in the inversion set of $\pi_1$. Then, for any $\pi$, we have that
$$\rho(\pi_2,\pi)-\rho(\pi_1,\pi) = \mathbb{I}((a_i,a_j) \in \text{inv}(\pi;\pi_\text{null})) = \rho(\tilde{\pi}_2,\pi)-\rho(\tilde{\pi}_1,\pi).$$
Furthermore, it is straightforward to check that $\rk(\tilde{\pi}_2) \geq j-i = \rk(\pi_2)$. We have thus shown that $\setS$ satisfies the first property in Definition~\ref*{defn:set_S}.

To show the second property, consider covering pairs $(\pi_1,\pi_2) \in \setS$ where the difference between the two inversion sets is the relation $(a_i,a_j)$. Let $(\pi_3,\pi_4) \in \setS$ where the difference between the two inversion sets is the relation $(a_k,a_l)$. By construction of the set $\setS$, $(a_i,a_j) \neq (a_k,a_l)$. Let $\pi$ be a permutation with $(a_i,a_j)$ in its inversion set. Then, as desired, 
$$\rho(\pi_2,\pi)-\rho(\pi_1,\pi) = \mathbb{I}((a_i,a_j) \in \text{inv}(\pi;\pi_\text{null})) \neq  \rho({\pi}_4,\pi)-\rho({\pi}_3,\pi).$$

\subsubsection{Characterizing $c_\Lp(\pi_1,\pi_2)$ for Covering Pair $(\pi_1,\pi_2)$} Since for any $\pi$, $\rho(\pi_2,\pi)-\rho(\pi_1,\pi) = \mathbb{I}((a_i,a_j) \in \text{inv}(\pi;\pi_\text{null}))$ for some pair of elements $(a_i,a_j)$, then $c_\Lp(\pi_1,\pi_2) = 1$.

\subsubsection{Refined False Discovery Bound for Total Ranking}
Let $\hat{\pi}_\text{stable}$ be output of Algorithm~\ref*{alg:poset_stability_discrete} with $\Psi = \Psi_\text{stable}$. Then: 
$$\mathbb{E}[\FD(\hat{\pi}_\text{stable},\pi^\star)] \leq \sum_{k = 1}^{p-1} \frac{q_k^2}{(1-2\alpha)(p-k)},$$
where 
$$q_k = \sum_{(\pi_1,\pi_2)\in\setS_k} \mathbb{E}[\rho(\pi_2,\hat{\pi}_\text{sub})-\rho(\pi_1,\hat{\pi}_\text{sub})] = \sum_{(i,j), j - i = k} \left[\mathbb{I}[(a_i,a_j) \in \text{inv}(\hat{\pi}_\text{sub};\pi_\text{null})\right].$$ 
Here, $\hat{\pi}_\text{sub}$ represents ranking from supplying $n/2$ samples to the base estimator. We can use the following data-driven approximation for $q_k$:$q_k \approx \frac{1}{B}\sum_{(i,j), j - i = k} \sum_{\ell=1}^B\left[\mathbb{I}[(a_i,a_j) \in \text{inv}(\hat{\pi}_\text{base}(\mathcal{D}^{(\ell)});\pi_\text{null})]\right]$, where $\hat{\pi}_\text{base}(\mathcal{D}^{(\ell)})$ represents the total ranking obtained by supplying the base estimator on dataset $\mathcal{D}^{(\ell)}$.
\subsection{Clustering}
We have a collection of $p$ items $\{a_1,a_2,\dots,a_p\}$ that we wish to cluster. We let $x_0 = \{\{a_1\},\{a_2\},\dots,\{a_p\}\}$ be the least element. As described in the main paper, will use the similarity valuation $\rho:= \rho_\text{meet}$ defined in $\mathrm{Eq.}$ (1) of the main paper. Since the clustering poset is meet semi-lattice,  $\rho$ computes the rank of the meet of two elements; in this setting, the
meet $x \wedge z$ of $x = \{G_1,\dots,G_q\}$ and $z = \{\tilde{G}_1,\dots,\tilde{G}_{s}\}$ is
$$x \wedge z = \{G_i \cap \tilde{G}_j: G_i \cap \tilde{G}_j \neq \emptyset\}.$$
Subsequently, $\rho(x,z) = \rk(x \wedge y)$ is $p - \# \text{ groups in }x \wedge z$, which can be equivalently expressed as:
$$\rho(x,z) = \sum_{i,j:|G_i \cap \tilde{G}_j| \neq \emptyset} |G_i \cap \tilde{G}_j|-1.$$
For sets $G_1,G_2 \subseteq \{1,2,\dots,p\}$ with $G_1 \cap G_2 = \emptyset$, we define:
$$\mathcal{R}_{G_1,G_2} :=\{\{a_1\},\{a_2\},\dots,\{a_p\}\} \setminus \{\{a_i\}:a_i \in G_1 \cup G_2\}.$$ 


\subsubsection{Characterizing $\setS$ for Clustering}
We construct a set $\setS$ satisfying the properties in Definition~\ref*{defn:set_S}. Initialize $\setS = \emptyset$. Then, for every $k = 1,2,\dots,p-1$ and pairs of groups of variables $G_1 \subseteq \{a_1,\dots,a_p\}$ and $G_2 \subseteq \{a_1,\dots,a_p\}$ with $|G_1|+|G_2| = k+1$ and $G_1 \cap G_2 = \emptyset$, we generate covering pairs $(x,y)$ with $y = \{G_1 \cup G_2,\mathcal{R}_{G_1,G_2}\}$ and $x = \{G_1,G_2,\mathcal{R}_{G_1,G_2}\}$, and let
$$S = S \cup (x,y).$$
Recalling that $\setS_k = \{(x,y) \in \setS, \rk(y) = k\}$, it is straightforward to check that for every $k = 1,2,\dots,p-1$
$$|\setS_k| = \begin{pmatrix}p \\k+1\end{pmatrix}\sum_{\ell=1}^k \begin{pmatrix}k+1 \\ l\end{pmatrix}.$$

Here, the terms $\big(\begin{smallmatrix}p \\ k+1 \end{smallmatrix}\big)$ counts the number of possible items in $G_1 \cup G_2$ and the term $\sum_{\ell=1}^{k+1}\big(\begin{smallmatrix}k+1 \\ l \end{smallmatrix}\big)$ counts the number of possible configurations of the group $G_2$. We will show that the constructed set $\setS$ satisfies Definition~\ref*{defn:set_S} of the main paper. Our analysis is based on the following lemma.

\setcounter{theorem}{9}
\begin{lemma} 
Consider the covering pairs $(x,y)$ with $x = \{G_1,G_2,\dots,G_q\}$ and $y = \{G_1 \cup G_2,G_3,\dots,G_q\}$ where $G_i \subseteq \{1,2,\dots,px\}$ and $G_i \cap G_j = \emptyset$ for every $i \neq j$. Let $(\tilde{x},\tilde{y})$ be covering pairs with $\tilde{y} = \{G_1 \cup G_2,\mathcal{R}_{G_1,G_2}\}$ and $\tilde{x} = \{G_1, G_2,\mathcal{R}_{G_1,G_2}\}$. Then, for every $z \in \Lp$, $\rho(y,z) - \rho(x,z) = \rho(\tilde{y},z)-\rho(\tilde{x},z)$. 
\label{lemma:clustering_red}
\end{lemma}
\begin{proof}[Proof of Lemma~\ref{lemma:clustering_red}]
Let $z = \{\tilde{G}_1,\dots,\tilde{G}_s\}$ with $\tilde{G}_i \subseteq \{a_1,a_2,\dots,a_p\}$ and $\tilde{G}_i \cap \tilde{G}_j = \emptyset$ for every $i \neq j$. Then:
$$\rho(y,z) = \sum_{j: (G_1 \cup G_2) \cap \tilde{G}_j \neq \emptyset}|(G_1 \cup G_2) \cap \tilde{G}_j|-1 + \sum_{i \geq 3, j: G_i \cap \tilde{G}_j \neq \emptyset}  |G_i \cap \tilde{G}_j|-1,$$
and
$$\rho(x,z) = \sum_{j: G_1 \cap \tilde{G}_j \neq \emptyset}|G_1 \cap \tilde{G}_j|-1 + \sum_{j: G_2 \cap \tilde{G}_j \neq \emptyset}|G_2 \cap \tilde{G}_j|-1+\sum_{i \geq 3, j: G_i \cap \tilde{G}_j \neq \emptyset}  |G_i \cap \tilde{G}_j|-1.$$
Since $\mathcal{R}_{G_1,G_2}$ consists of groups of size one, we have that:
$$\rho(\tilde{y},z) =  \sum_{j: (G_1 \cup G_2) \cap \tilde{G}_j \neq \emptyset}|(G_1 \cup G_2) \cap \tilde{G}_j|-1,$$
and 
$$\rho(\tilde{x},z) =   \sum_{j: G_1 \cap \tilde{G}_j \neq \emptyset}|G_1 \cap \tilde{G}_j|-1 + \sum_{j: G_2 \cap \tilde{G}_j \neq \emptyset}|G_2 \cap \tilde{G}_j|-1.$$
We thus can see that $\rho(y,z) - \rho(x,z) = \rho(\tilde{y},z) - \rho(\tilde{x},z)$.
\end{proof}

\noindent \underline{Showing $\setS$ satisfies Definition~\ref*{defn:set_S}} With Lemma~\ref{lemma:clustering_red} at hand, we show that out constructed $\setS$ satisfies Definition~\ref*{defn:set_S} of the main paper. We start with the first property. Consider any $(u',v') \subseteq \Lp$. Without loss of generality, we take $v' = \{G_1\cup G_2,G_3,\dots,G_q\}$ and $u' = \{G_1,G_2,\dots,G_q\}$. We let $v = \{G_1 \cup G_2,\mathcal{R}_{G_1,G_2}\}$ and $u = \{G_1,G_2,\mathcal{R}_{G_1,G_2}\}$. Then, according to Lemma~\ref{lemma:clustering_red}, we have that $\rho(v',z)-\rho(u',z) = \rho(v,z)-\rho(u,z)$. Furthermore, since $\rk(x) = p - \# \text{ groups in } x$, we have that $\rk(v) \leq \rk(v')$. Thus, the first property of $\setS$ is satisfied. We demonstrate the second property. Consider any $(u,v) \in \setS$ and $(u',v') \in \setS$ that are different. Let $u = \{G_1,G_2,\mathcal{R}_{G_1,G_2}\}$ and $v = \{G_1 \cup G_2,\mathcal{R}_{G_1,G_2}\}$. Additionally, let $u' = \{G_1',G_2',\mathcal{R}_{G_1',G_2'}\}$ and $v' = \{G'_1 \cup G'_2,\mathcal{R}_{G'_1,G'_2}\}$. Since the covering pairs $(u,v)$ and $(u',v')$ are different, there must exist two items $a_i,a_j$ such that either $(a_i,a_j)$ are grouped together in $v$ but are not together in $u$ or $(a_i,a_j)$ are grouped together in $v'$ but are not together in $u'$. 
Let $z = \{\{a_i,a_j\},\mathcal{R}_{\{a_i\},\{a_j\}}\}$. Since $\rho(v,z)-\rho(u,z) = \mathbb{I}[(a_i,a_j) \text{ grouped together in }v \text{ but not in }u]$ and $\rho(v',z)-\rho(u',z) = \mathbb{I}[(a_i,a_j) \text{ grouped together in }v' \text{ but not in }u']$, we have that $\rho(v,z)-\rho(u,z) \neq \rho(v',z)-\rho(u',z)$. 
\subsubsection{Characterizing $c_\Lp(u,v)$ for Covering Pair $(u,v)$} 
\begin{lemma}Let $v = \{G_1\cup{G}_2,\mathcal{R}_{G_1,G_2}\}$ and $u =\{G_1,{G}_2,\mathcal{R}_{G_1,G_2}\}$ be a covering pair $(u,v)\in\setS$. Then, $c_\Lp(u,v) = \min\{|G_1|,|G_2|\}$. 
\label{lemma:c_cluster}
\end{lemma}
\begin{proof}[Proof of Lemma~\ref{lemma:c_cluster}]
Let $z = \{\tilde{G}_1,\dots,\tilde{G}_q\}$. Then, from proof of Lemma~\ref{lemma:clustering_red}, we have that:
$$\rho(v,z)-\rho(u,z) = \left[\sum_{j: (G_1 \cup G_2) \cap \tilde{G}_j \neq \emptyset}|(G_1 \cup G_2) \cap \tilde{G}_j|-1\right] - \left[\sum_{j: G_1 \cap \tilde{G}_j \neq \emptyset}|G_1 \cap \tilde{G}_j|-1\right]- \left[\sum_{j: G_2 \cap \tilde{G}_j \neq \emptyset}|G_2 \cap \tilde{G}_j|-1\right].$$
Let $I_1 := \{j:\tilde{G}_j \cap G_1 \neq \emptyset\}$ and  $I_2 := \{j:\tilde{G}_j \cap G_2 \neq \emptyset\}$. Then, 
$$\rho(v,z)-\rho(u,z) = \left[\sum_{j \in I_1 \cup I_2}|(G_1 \cup G_2) \cap \tilde{G}_j|-1\right]  - \left[\sum_{j\in{I}_1}|G_1 \cap \tilde{G}_j|-1\right]- \left[\sum_{j \in I_2}|G_2 \cap \tilde{G}_j|-1\right].$$
Simple manipulations yield:
$$\rho(v,z)-\rho(u,z) = \left[\sum_{j \in I_1 \cap I_2}|(G_1 \cup G_2) \cap \tilde{G}_j|-1\right]  - \left[\sum_{j\in{I}_1 \cap I_2}|G_1 \cap \tilde{G}_j|-1\right]- \left[\sum_{j \in I_1 \cap I_2}|G_2 \cap \tilde{G}_j|-1\right].$$

Clearly, if $I_1 \cap I_2 = \emptyset$, then $\rho(v,z)-\rho(u,z) = 0$. Suppose  $I_1 \cap I_2 \neq \emptyset$. Then, 
$$\rho(v,z)-\rho(u,z) = |I_1 \cap I_2|+\left[\sum_{j \in I_1 \cap I_2}|(G_1 \cup G_2) \cap \tilde{G}_j|-|G_1 \cap \tilde{G}_j|- |G_2 \cap \tilde{G}_j|\right] = |I_1 \cap I_2|.$$
Notice that $|I_1 \cap I_2| \leq \min\{|G_1|,|G_2|\}$. Then, the upper bound can be achieved by for example setting $z = \{N,\{\{a_1\},\{a_2\},\dots,\{a_p\}\setminus{N}\}$ with $N = \{(a_i,a_j): a_i \in G_1, a_j \in G_2\}$.  
\end{proof}
\subsubsection{Refined False Discovery Bound for Clustering} Let $\hat{x}_\text{stable}$ be output of Algorithm~\ref*{alg:poset_stability_discrete} with $\Psi = \Psi_\text{stable}$. Then: 
$$\E[\FD(\hat{x}_\text{stable},x^\star)] \leq \sum_{k = 1}^{p-1}\frac{q_k^2}{(1-2\alpha)\begin{pmatrix}p\\ k+1\end{pmatrix}\sum_{\ell=1}^k\begin{pmatrix}k+1\\l\end{pmatrix}},$$
where,
\begin{eqnarray*}
    \begin{aligned}
q_k &= \sum_{(u,v)\in\setS_k} \frac{\mathbb{E}[\rho(v,\hat{x}_\text{sub})-\rho(u,\hat{x}_\text{sub})]}{c(u,v)}\\ &= \sum_{\substack{G_1 \subseteq \{a_1,\dots,a_p\}, G_2 \subseteq \{a_1,\dots,a_p\}\\ G_1\cap G_2 = \emptyset; |G_1|+|G_2| = k+1}} \frac{\mathbb{E}[\# \text{ groups }\hat{G}_j \text{ in }\hat{x}_\text{sub} \text{ satisfying } \hat{G}_j \cap G_1 \neq \emptyset \text{ and }\hat{G}_j \cap G_2 \neq \emptyset]}{\min\{|G_1|,|G_2|\}}.
\end{aligned}
\end{eqnarray*}
Here, $\hat{x}_\text{sub}$ represents clustering from supplying $n/2$ samples to the base estimator. We will use the following data-driven approximation to estimate $q_k$
$$q_k \approx \frac{1}{B}\sum_{\substack{G_1 \subseteq \{a_1,\dots,a_p\}, G_2 \subseteq \{a_1,\dots,a_p\}\\ G_1\cap G_2 = \emptyset; |G_1|+|G_2| = k+1}} \sum_{\ell=1}^B\frac{\# \text{ groups }\hat{G}_j \text{ in }\hat{x}_\text{base}(\mathcal{D}^{(\ell)}) \text{ satisfying } \hat{G}_j \cap G_1 \neq \emptyset \text{ and }\hat{G}_j \cap G_2 \neq \emptyset]}{\min\{|G_1|,|G_2|\}},$$
with $\hat{x}_\text{base}(\mathcal{D}^{(\ell)})$ represents the partition obtained from supplying $\mathcal{D}^{(\ell)}$ to the base estimator.

\subsection{Causal Structure Learning}
Throughout, we consider covering pairs $(\mathcal{C}_u,\mathcal{C}_v)$ where each connected component in 
the skeletons of $\mathcal{C}_u,\mathcal{C}_v$ have a diameter at most two. We denote this set by $\mathcal{T}$. Note that for any covering pair $(\mathcal{C}_u,\mathcal{C}_v) \in \mathcal{T}$, $\mathcal{C}_v$ is a polytree (no cycles after removing direction from all edges). Throughout, we will use the similarity valuation $\rho:= \rho_\text{meet}$. Our analysis in this section will build on the following proposition.

\begin{proposition}Let $\mathcal{C}_u$ and $\mathcal{C}_v$ be two CPDAGs that are polytrees with $\cc_u \preceq \cc_v$. Then, the following statements hold:
\begin{enumerate}
    \item[(a)] for any pairs of nodes $\mathcal{E}$, the set of DAGs that result from removing edges among pairs $\mathcal{E}$ in any DAG $\cg_v$ form a Markov equivalence class.
    \item[(b)] for every DAG $\mathcal{G}_v \in \mathcal{C}_v$, there exists a DAG $\mathcal{G}_u \in \mathcal{C}_u$ such that $\mathcal{G}_u$ is a directed subgraph of $\mathcal{G}_v$. 
\end{enumerate}
\label{lemma:causal}
\end{proposition}
The proof of this proposition relies on the following lemmas. 

{
\begin{lemma}Let $(\mathcal{C}_u,\mathcal{C}_v)$ be a covering pair CPDAGs. Then, there must exist a $\cg_u \in \cc_u$ and a DAG $\cg_v \in \cc_v$ such that $\cg_u$ is a subgraph of $\cg_v$.
\label{lemma:causal_1}
\end{lemma}
\begin{proof} Let $\tilde{\cg}_u,\tilde{\cg}_v$ be a pair of DAGs contained in $\cc_u$ and $\cc_v$, respectively. The conditional dependencies of $\tilde{\cg}_u$ are thus contained in $\tilde{\cg}_v$. Moreover, since $(\mathcal{C}_u,\mathcal{C}_v)$ are covering pairs, there exists only a single pair of nodes that are connected in $\tilde{\cg}_v$ and are not connected in $\tilde{\cg}_u$. By Proposition~\ref{prop:Chickering}, there exists a sequence of allowed edge reversals $R$ (that are not changing conditional dependencies) and edge deletions $D$ that apply to go from $\tilde{\cg}_v$ to $\tilde{\cg}_v$, e.g. $\tilde{\cg}_v + RRRRR+D+RRRR \to \tilde{\cg}_u$. Let $\cg_v$ be the DAG after the operation $\tilde{\cg}_v + RRRRR$. Let $\cg_u$ be the DAG after the deletion step applied to $\cg_v$. Notice that $\tilde{\cg}_v$ is in the same Markov equivalence class as $\cg_v$, and $\tilde{\cg}_u$ is in the same Markov equivalence class as $\cg_u$. 
\end{proof}

\begin{lemma}Let $\cc_u,\cc_v$ be a pair of CPDAGs with $\cc_u \preceq \cc_v$. Then, the skeleton of $\cc_u$ must be a subgraph of the skeleton of $\cc_v$ and every v-structure $i\rightarrow{k}\leftarrow{j}$ in $\cc_v$ must either also be present in $\cc_u$ or at least one of the edges between the pair of nodes $(i,k)$ or $(j,k)$ is missing. 
\label{lemma:causal_2}
\end{lemma}
\begin{proof} Let $X \in \mathbb{R}^p$ be the collection of random variables. To prove the first part, suppose as a point of contradiction that the pair of nodes $(i,j)$ are connected in $\cc_u$ and not in $\cc_v$. Then, in the model $\cc_u$, $X_i \not\perp X_j|X_C$ for all $C \subseteq \{1,2,\dots,p\} \setminus\{i,j\}$. On the other hand, in the model $\cc_v$, $X_i \perp X_j|X_C$ for some $C\subseteq \{1,2,\dots,p\} \setminus\{i,j\}$. This leads to a contradiction as the conditional dependencies of $\cc_u$ are not contained in $\cc_v$. To prove the second part, we note from the first part that the skeleton of $\cc_u$ must be contained in that of $\cc_v$. As a point of contradiction, suppose the pairs of nodes $(i,k)$ and $(j,k)$ are connected in $\cc_u$ without forming a v-structure, and are connected as $i \rightarrow k \leftarrow j$ in $\cc_v$ with $i$ not connected to $j$. Then, in the model $\cc_v$, $X_i \perp X_j | X_C$ for some subset $C = \{1,2,\dots,p\}\setminus \{i,j,k\}$. On the other hand, in the model $\cc_u$, $X_i \not\perp X_j|X_C$ for all all subsets $C = \{1,2,\dots,p\}\setminus \{i,j,k\}$. This again leads to a contradiction as the conditional dependencies of $\cc_u$ are not contained in $\cc_v$.
\end{proof}

\begin{lemma}Let $\cc_u,\cc_v$ be a pair of CPDAGs that are polytrees with $\cc_u \preceq \cc_v$. Then, every v-structure in $\cc_u$ must also be a v-structure in $\cc_v$.
\label{lemma:causal_3}
\end{lemma}
\begin{proof}
Suppose $i \rightarrow k \leftarrow j$ in $\cc_u$. Then, by Lemma~\ref{lemma:causal_2}, pairs of nodes $(i,j)$ and $(j,k)$ must be connected in $\cc_v$. Furthermore, since $\cc_v$ is a polytree, $(i,j)$ cannot be connected. Suppose as a point of contradiction that the pairs $(i,j)$ and $(j,k)$ do not form a v-structure in $\cc_v$. Then, in the model $\cc_u$, $X_i \not\perp X_j | X_C$ for all subsets $C$ that contain $k$, whereas in the model $\cc_2$, $X_i \perp X_j | X_C$ for some subset $C$ containing $k$. This leads to a contradiction as the conditional dependencies of $\cc_u$ are not contained in $\cc_v$.
\end{proof}
}

\begin{proof}[Proof of Proposition~\ref{lemma:causal}]
We first prove part (a).  By the polytree assumption, it follows that for any DAG $\cg_v$ in the CPDAG $\mathcal{C}_v$, removing the edges among pairs in $\mathcal{E}$ does not create any v-structures, and removes the same (potentially empty) v-structures. That means that the collection of DAGs obtained by taking any DAG in $\cc_v$ and removing the edges between the pairs of nodes $\mathcal{E}$ will have the same skeleton and same v-structures, and are thus in the same Markov equivalence class. 

We next prove part (b). Let $\mathcal{E} = \{(i,j)\}$ be collection of nodes that connected in $\mathcal{C}_v$ but not in $\mathcal{C}_u$. Since $\mathcal{C}_v$ is a polytree, by Lemma~\ref{lemma:causal_3}, any v-structure in $\mathcal{C}_u$ is also present in $\mathcal{C}_v$. Furthermore, by Lemma~\ref{lemma:causal_2}, every v-structure $i \rightarrow k \leftarrow j$  in $\mathcal{C}_v$ is either present in $\mathcal{C}_v$ or at least one of the edges between pair of nodes $(i,k)$ and $(j,k)$ is missing. 
Consider any DAG $\mathcal{G}_v$ in $\mathcal{C}_v$. Let $\mathcal{G}_u$ be the DAG that is obtained by deleting the edges in $\mathcal{G}_v$ among pairs of nodes $\mathcal{E}$. By construction, $\mathcal{G}_u$ has the same skeleton as any DAG in $\mathcal{C}_u$. Since $\mathcal{C}_v$ is a polytree, every structure present in $\mathcal{G}_u$ or in $\mathcal{C}_u$ are also in $\mathcal{C}_v$. Suppose as a point of contradiction that there exists a v-structure $i \rightarrow k \leftarrow j$ that is in $\mathcal{G}_u$ and not in $\mathcal{C}_u$. Since this v-structure must also be present in $\mathcal{G}_v$, and from the previous statement, then one of the edges between pairs of nodes $(i,k)$ or $(j,k)$ must be missing in $\mathcal{C}_u$ which contradicts $\mathcal{C}_u$ and $\mathcal{G}_u$ having same skeleton. A similar argument allows us to conclude that this leads to a contradiction. Thus, $\mathcal{G}_u$ and $\mathcal{C}_u$ have same skeleton and v-structures. We conclude that $\mathcal{G}_u$ must be a DAG in $\mathcal{C}_u$.  
\end{proof}

\subsubsection{Characterizing $\setS$ for Causal Structure Learning}
We construct the set $\setS$ as follows. Initialize $\setS = \emptyset$. For every reference node, and $k = 1,\dots,p-1$, let $\mathcal{C}_y$ be a CPDAG generated with $k$ edges, where every edge is between the reference node and another node; no other edges can be added without violating the condition that the largest undirected path has size less than or equal to two. A consequence of Proposition~\ref{lemma:causal} is that there are $k$ CPDAGs $\mathcal{C}_{x_1},\dots,\mathcal{C}_{x_k}$ that form a covering pair with $\cc_y$. We then let $$\setS= \setS \cup (\mathcal{C}_{x_i},\mathcal{C}_{y}),$$ for every $i = 1,2,\dots,k$. Recall that $\setS_k := \{(\cc_x,\cc_y) \in \setS, \rk(\cc_y) = k\}$. Then, {$|\mathcal{S}_1| = \begin{pmatrix}p\\2\end{pmatrix}$} and for $k > 1$
{$$|\setS_k| = pk\begin{pmatrix}p-1 \\k\end{pmatrix}\sum_{i \in \{0,2\dots,k\}}\begin{pmatrix}k \\i\end{pmatrix}.$$}
The result above follows from noting that for every reference node and $k$ other nodes, there are $\sum_{i \in \{0,2\dots,k\}}\big(\begin{smallmatrix}k \\i\end{smallmatrix}\big)$ possible CPDAGs that are polytrees can formed by connecting the $k$ nodes to the reference node; the factor $p\big(\begin{smallmatrix}p-1 \\k\end{smallmatrix}\big)$ comes from $p$ total possible reference nodes and $\big(\begin{smallmatrix}p-1 \\k\end{smallmatrix}\big)$ possible set of $k$ nodes to connect to the reference node.\\

\noindent We will show that the constructed set $\setS$ satisfies Definition~\ref*{defn:set_S} of the main paper. Our analysis is based on the following lemma.
\begin{lemma}Let $\mathcal{C}_{\tilde{y}}$ be a CPDAG that is a polytree and contains $m$ disconnected subgraphs (both directed and undirected). Let $\mathcal{C}_{\tilde{y}_{i}}$ be each disconnected subgraph for $i = 1,2,\dots,m$. Then, for any CPDAG $\mathcal{C}_z$,
$$\rho(\mathcal{C}_{\tilde{y}},\mathcal{C}_z) = \sum_{i = 1}^m\rho(\mathcal{C}_{\tilde{y}_{i}},\mathcal{C}_z).$$
\label{lemma:causal_2}
\end{lemma}
\begin{proof} 
    We will first show that $\rho(\mathcal{C}_{\tilde{y}},\mathcal{C}_z) \leq \sum_{i = 1}^m\rho(\mathcal{C}_{\tilde{y}_i},\mathcal{C}_z)$. Let $\mathcal{C}_{\tilde{x}} \in \argmax_{\mathcal{C}_x\preceq \mathcal{C}_{\tilde{y}}, \cc_x \preceq \mathcal{C}_z} \rk(\mathcal{C}_x)$. By Proposition~\ref{lemma:causal}, if $\mathcal{C}_x \preceq \mathcal{C}_{\tilde{y}}$, there is a DAG $\cg_x$ in $\mathcal{C}_x$ and a DAG $\cg_{\tilde{y}}$ in $\mathcal{C}_{\tilde{y}}$ such that $\cg_x$ is a subgraph of $\cg_{\tilde{y}}$. Since $\cg_{\tilde{y}}$ has disconnected components, so must $\cg_x$.  We let $\mathcal{C}_{\tilde{x}_{i}}$ be the subgraphs of $\mathcal{C}_{\tilde{x}}$ where every subgraph $\mathcal{C}_{\tilde{x}_{i}}$ only contains edges among nodes that are connected (to other nodes) in the graph $\mathcal{C}_{\tilde{y}_{i}}$. By construction,  $\cc_{\tilde{x}_i} \preceq \cc_{\tilde{y}_i}$, $\rk(\mathcal{C}_{\tilde{x}}) = \sum_{i = 1}^{m}\rk(\mathcal{C}_{\tilde{x}_{i}})$, and $\mathcal{C}_{\tilde{x}_{i}}\preceq \mathcal{C}_z$. Thus, $\rk(\cc_{\tilde{x}_i}) \leq \rho(\cc_{\tilde{y}_i},\cc_z)$. Then, we can conclude that 
    $$\sum_{i=1}^m \rho(\mathcal{C}_{\tilde{y}_{i}},\mathcal{C}_z) \geq \sum_{i=1}^m \rk(\mathcal{C}_{\tilde{x}_{i}}) = \rk(\cc_{\tilde{x}})= \rho(\mathcal{C}_{\tilde{y}},\mathcal{C}_z).$$ 

    Now we will show that $\rho(\mathcal{C}_{\tilde{y}},\mathcal{C}_z) \geq \sum_{i = 1}^m\rho(\mathcal{C}_{\tilde{y}_{i}},\mathcal{C}_z)$. Let $\mathcal{C}_{\tilde{x}_{i}} \in \argmax_{\mathcal{C}_x \preceq \mathcal{C}_{\tilde{y}_{i}},\cc_x \preceq \mathcal{C}_z}\rk(\mathcal{C}_x)$. Now form a CPDAG $\mathcal{C}_{\bar{y}}$ by combining all the disjoint graphs $\mathcal{C}_{\tilde{x}_{i}}$ for every $i = 1,2,\dots,m$ into one graph. Since these graphs are disjoint (i.e. nodes that are connected in each graph are distinct), we have that $\mathcal{C}_{\bar{y}} \preceq \mathcal{C}_{\tilde{y}}$ and $\mathcal{C}_{\bar{y}} \preceq \mathcal{C}_z$ and that $\rk(\cc_{\bar{y}}) = \sum_{i=1}^m \rk(\cc_{\tilde{x}_i})$. So we conclude that $$\rho(\mathcal{C}_{\tilde{y}},\mathcal{C}_z) \geq \rk(\mathcal{C}_{\bar{y}}) =  \sum_{i=1}^m \rk(\cc_{\tilde{x}_i}) = \sum_{i = 1}^m\rho(\mathcal{C}_{\tilde{y}_{i}},\mathcal{C}_z).$$    
\end{proof}
\underline{Showing $\setS$ satisfies Definition~\ref*{defn:set_S}}
For the first property, consider covering pairs $(\mathcal{C}_{u'},\mathcal{C}_{v'}) \in T$. Let $(i,j)$ be the pair of nodes that are connected in $\cc_{v'}$ and are not connected in $\cc_{u'}$. Since every undirected path in $\cc_{v'}$ has size at most $2$, then $\cc_{v'}$ decouples into two disconnected CPDAGs $\cc_{v}$ and $\cc_{1}$, where $\cc_{v}$ only involves nodes adjacent to $(i,j)$. Similarly, $\cc_{u'}$ decouples into two disconnected CPDAGs $\cc_{u}$ and $\cc_{2}$, where $\cc_{2} = \cc_{1}$ and $\cc_{u}$ is covered by $\cc_{v}$. From Lemma~\ref{lemma:causal_2}, we have that for any CPDAG $\cc_z$
$$ \rho(\cc_{v'},\cc_z)-\rho(\cc_{u'},\cc_z) = \rho(\cc_v,\cc_z)-\rho(\cc_u,\cc_z).$$
Notice that $(\cc_u,\cc_v) \in \setS$. Furthermore, since the number of edges (directed and undirected) in $\cc_{v'}$ is larger than $\cc_{v}$, we have that $\rk(\cc_v) \leq \rk(\cc_{v'})$.  

We next show the second property in Definition~\ref*{defn:set_S}. Let $(\cc_u,\cc_v) \in \setS$ and $(\cc_{u'},\cc_{v'}) \in \setS$. Our objective is to show that $\rho(\cc_v,\cc_z)-\rho(\cc_u,\cc_z) = \rho(\cc_{v'},\cc_{z})-\rho(\cc_{u'},\cc_{z})$ for all $\cc_z \Leftrightarrow \cc_u = \cc_{u'}$ and $\cc_v = \cc_{v'}$. The direction $\leftarrow$ trivially holds, and hence we focus on the direction $\rightarrow$. We consider multiple scenarios; throughout the extra edge that is present in $\cc_{v}$ and not in $\cc_{u}$ is between the pair of nodes $(i,j)$, and the extra edge that is present in $\cc_{v'}$ and not in $\cc_{u'}$ is between the pair of nodes $(k,l)$.
\begin{itemize}
  \item[(1)] Suppose that the nodes $(k,l)$ are not connected in $\cc_v$. Letting $\cc_z$ be a CPDAG with only an edge between nodes $(k,l)$, we find that $\rho(\cc_v,\cc_z)-\rho(\cc_u,\cc_z) = 0$ and $\rho(\cc_{v'},\cc_{z})-\rho(\cc_{u'},\cc_{z}) = 1$. So this scenario cannot occur.
    \item[(2)] Suppose there is an edge between pairs $(s,t)$ in $\cc_{u'}$ that is missing in $\cc_{v}$ (and as a result in $\cc_{u}$). Construct CPDAG $\cc_z$ with two edges, one between the pair $(i,j)$ and another between the pair $(s,t)$ with the property that $\cc_z \not\preceq \cc_{v'}$; this construction is possible since $(\cc_{u'},\cc_{v'})\in \setS$, meaning that if there is an edge between pair of nodes $(i,j)$ in $\cc_{v'}$, this edge is incident to the edge between the pair of nodes $(s,t)$. Then, it is evident that $\rho(\cc_v,\cc_z)-\rho(\cc_u,\cc_z) = 1$ but $\rho(\cc_{v'},\cc_z)-\rho(\cc_{u'},\cc_z) = 0$. So this scenario cannot occur.   
    \item[(3)] Suppose there is an edge between pairs $(s,t)$ in $\cc_{u'}$ that is missing in $\cc_{u}$ but is not missing in $\cc_{v}$. Let $\cc_z$ be a CPDAG only containing an edge between $(s,t)$. Then it follows that $\rho(\cc_v,\cc_z)-\rho(\cc_u,\cc_z) = 1$ but $\rho(\cc_{v'},\cc_z)-\rho(\cc_{u'},\cc_z) = 0$. So this scenario cannot occur.
    \end{itemize}
    From the impossibilities of scenarios 1-2, and noting that a similar argument can be made by swapping $\cc_{u'}$ with $\cc_{u}$, and $\cc_{v'}$ with $\cc_v$, we conclude that $\cc_v,\cc_{v'}$ have edges between the same pairs of nodes. Combining this result with the impossibility of scenario 3, we conclude that $\cc_{u},\cc_{u'}$ have edges between the same pairs of nodes. We then continue with the final scenario.
    \begin{itemize}
    \item[(4)] Suppose that $\cc_{v}$ and $\cc_{v'}$ are not identical CPDAGs. Since both $\cc_{v}$ and $\cc_{v'}$ have maximum undirected path length less than or equal to two, they both must have the same reference node $i$ (where the other nodes are connected to). Furthermore, since $\cc_{v}$ and $\cc_{v'}$ have the same skeleton and are different, they must have strictly more than one edge, and they must have different v-structures. As a first sub-case, suppose $\cc_{v'}$ have a v-structure $s \rightarrow i \leftarrow t$ that is not present in $\cc_v$, so that $s \leftarrow i$ or $s - i$ in $\cc_v$. Then, let $\cc_z$ be a CPDAG containing two edges between the pairs $(i,j)$ and $(i,s)$ with $\cc_z \preceq\cc_{v}$. By construction, $\rho(\cc_v,\cc_z) -\rho(\cc_u,\cc_z)=1$ but $\rho(\cc_{v'},\cc_z)-\rho(\cc_{u'},\cc_z) = 0$. Swapping $\cc_{u'}$ with $\cc_{u}$, and $\cc_{v'}$ with $\cc_v$, and following similar arguments, we arrive again at a contradiction if $\cc_{v}$ has a v-structure that is not present in $\cc_{v'}$. 
\end{itemize}
From the impossibility of scenario 4, we conclude that $\cc_{v}$ and $\cc_{v'}$ have the same skeleton and v-structure and consequently $\cc_{v}=\cc_{v'}$. We thus have that $\cc_{u} \preceq \cc_v$ and $\cc_{u'} \preceq \cc_v$. Furthermore, since $\cc_{u'}$ and $\cc_{u}$ have the same skeleton, both are missing an edge between pair of nodes $(i,j)$ that is connected in $\cc_v$. Appealing to part a of Proposition~\ref{lemma:causal}, we conclude that $\cc_{u} = \cc_{u'}$. 

\subsubsection{Characterizing $c_\Lp(\cc_u,\cc_v)$ for Covering Pairs $(\cc_u,\cc_v)$} We have the following lemma.
\begin{lemma} Let $(\cc_u,\cc_v)$ be CPDAGs that are polytrees and form a covering pair. Then, $c_\Lp(\cc_u,\cc_v) = 1$.
\end{lemma}
\begin{proof}
Let the pair of nodes $(i,j)$ be connected in $\cc_v$ and not connected in $\cc_u$. Consider any CPDAG $\cc_z$. Let $\cc_{\tilde{y}}\in \argmax_{\cc_y \preceq \cc_v, \cc_y\preceq \cc_z} \rk(\cc_y)$. Since the CPDAG $\cc_v$ is a polytree, so is the CPDAG $\cc_{\tilde{y}}$. Let $\cg_v$ be any DAG in $\cc_v$. Then, by part b of Proposition~\ref{lemma:causal}, there exists DAGs $\cg_{\tilde{y}}^{(1)} \in \cc_{\tilde{y}}$ and $\cg_{u} \in \cc_u$ such that $\cg_{\tilde{y}}^{(1)}$ and $\cg_{u}$ are both subgraphs of $\cg_v$. Suppose we remove an edge that may be present between the pair of nodes $(i,j)$ in $\cg_{\tilde{y}}$ and denote the resulting subgraph by $\cg_x$. By construction, $\cg_x$ is also a subgraph of $\cg_{u}$. Let $\mathcal{C}_{x}$ be the CPDAG by completing the DAG $\cg_{u}$. Then, $\mathcal{C}_x \preceq \mathcal{C}_u$. Since $\cg_x$ is a subgraph of $\cc_{\tilde{y}}$, $\mathcal{C}_x \preceq \mathcal{C}_{\tilde{y}}$, and by transitivity, $\mathcal{C}_x \preceq \mathcal{C}_z$. Consequently, $\rk(\mathcal{C}_x) \geq \rk(\cc_{\tilde{y}})$, and thus $\rho(\cc_v,\cc_z)-\rho(\cc_u,\cc_z) \leq 1$.
\end{proof}

\subsubsection{Refined False Discovery Bound for Causal Structure Learning} Let $\hat{\cc}_\text{stable}$ be output of Algorithm~\ref*{alg:poset_stability_discrete} with $\Psi = \Psi_\text{stable}$. Let $\mathcal{C}^\star$ be the population CPDAG. Then: 
$$\E[\FD(\hat{\cc}_\text{stable},\cc^\star)] \leq \frac{q_1^2}{(1-2\alpha) \begin{pmatrix}p \\2\end{pmatrix}}+ \sum_{k = 2}^{p-1}\frac{q_k^2}{(1-2\alpha) p\begin{pmatrix}p-1 \\k-1\end{pmatrix}\sum_{i \in \{0,2\dots,k\}}\begin{pmatrix}k \\i\end{pmatrix}},$$
where,
$$q_k = \sum_{(\cc_u,\cc_v)\in\setS_k} \mathbb{E}[\rho(\cc_v,\hat{\cc}_{\text{sub}})-\rho(\cc_u,\hat{\cc}_\text{sub})].$$
Here, $\hat{\cc}_\text{sub}$ represents the CPDAG from supplying $n/2$ samples to the base estimator. We will use the following data-driven approximation to estimate $q_k$
$$q_k \approx \frac{1}{B}\sum_{\ell=1}^B\sum_{(\cc_u,\cc_v)\in\setS_k} {\mathbb{E}[\rho(\cc_v,\hat{\cc}_\text{base}(\mathcal{D}^{(\ell)})-\rho(\cc_u,\hat{\cc}_\text{base}(\mathcal{D}^{(\ell)}))]},$$
with $\hat{\cc}_\text{base}(\mathcal{D}^{(\ell)})$ represents the CPDAGs obtained from supplying dataset $\mathcal{D}^{(\ell)}$ to base estimator $\hat{\cc}_\text{base}$.

\section{Assumptions 1 and 2 of the Main Paper for the Total Ranking Problem in Example~\ref*{ex:complete-ranking}}
\label{sec:assumptions_total_ranking}
Let $S = \{a_1,a_2,\dots,a_p\}$ be the set of $p$ elements. Let $\pi_\text{null}(a_i) = i$ for every $i = 1,2,\dots,p$. We use the similarity valuation $\rho:= \rho_\text{total-ranking}$ in $\mathrm{Eq.}$ (2) of the main paper. As each element in the poset corresponds to a function $\pi : S \to S$, we will use this functional notation throughout. For a covering pair $(\pi_1,\pi_2)$, there exists a single pair of elements $(a_i,a_j) \in \text{inv}(\pi_2;\pi_\text{null})\setminus \text{inv}(\pi_1;\pi_\text{null})$ with $j>i$. Then, from the definition of $\rho$, for any permutation $\pi$, we have that $$\rho(\pi_2,\pi)-\rho(\pi_1,\pi) = \mathbb{I}[(a_i,a_j) \in \text{inv}(\pi;\pi_\text{null})] = \mathbb{I}[\pi(a_j) < \pi(a_i)].$$    

Let $\hat{\pi}_\text{sub}$ be the estimated ranking from applying a base procedure on a subsample of the data. Consider a fixed integer $k$ with $1\leq k\leq p-1$. Define the sets $S_1$ and $S_2$:
\begin{eqnarray*}
\begin{aligned}
S_1 &= \{(a_i,a_j) \in \text{inv}(\pi^\star;\pi_{\text{null}}):j -i = k\},\\
S_2 &= \{(a_i,a_j) \not\in \text{inv}(\pi^\star;\pi_{\text{null}}):j -i = k\}.\\
\end{aligned}
\end{eqnarray*}
The set $S_1$ corresponds to non-null pairs (as described in the main paper) and the set $S_2$ corresponds to null pairs. 

Then, appealing to the definition of $\setS$ and the constant $c_\mathcal{L}(\cdot,\cdot)$ in the total ranking case (see Section~\ref{sec:total_ranking_setS}), Assumption 1 of the main paper reduces to the following inequality being satisfied
\begin{equation}
\frac{\sum_{(a_i,a_j)\in S_1} \mathbb{P}(\hat{\pi}_\text{sub}(a_j)<\hat{\pi}_\text{sub}(a_i))}{\sum_{(a_i,a_j)\in S_2} \mathbb{P}(\hat{\pi}_\text{sub}(a_j)<\hat{\pi}_\text{sub}(a_i))} \geq \frac{|S_1|}{|S_2|}.
\label{eqn:assumpt_1_total_ranking_new}
\end{equation} 
Consider an estimator $\hat{\pi}_\text{sub} = \hat{\pi}_\text{random}$ that randomly selects a total ranking in the space of permutations. Then, for every $i$ and $j$, $\mathbb{P}(\hat{\pi}_\text{sub}(a_j)<\hat{\pi}_\text{sub}(a_i)) = \frac{1}{2}$. Thus, in this case, Assumption 1 in \eqref{eqn:assumpt_1_total_ranking_new} is satisfied with equality. 

It is also straightforward to check that Assumption 2 of the main paper is reduced to
$$\mathbb{P}(\hat{\pi}_\text{sub}(a_j) < \hat{\pi}_\text{sub}(a_i)) \text{ being the same for every }(a_j,a_i) \in S_2.$$

\section{Analysis in the Continuous Examples~\ref*{ex:subspace-estimation} and \ref*{ex:blind-source-separation}}
\label{sec:proof_general_stability}
For notational ease, we let $\hat{x}_\text{base}^{(\ell)} = \hat{x}_\text{base}(\mathcal{D}^{(\ell)})$. Notice that for any $l = 1,2,\dots,B$:
\begin{equation*}
\begin{aligned}
    \mathrm{FD}(\hat{x}_\text{stable},x^\star) &= \rk(\hat{x}_\text{stable}) -\rho(\hat{x}_\text{stable},x^\star)\\
    &= \left[\rk(\hat{x}_\text{stable})-\rho(\hat{x}_\text{stable},\hatx_\text{base}^{(\ell)})\right] + \left[\rk(\hatx_\text{base}^{(\ell)})-\rho(\hatx_\text{base}^{(\ell)},x^\star)\right] + \kappa(\hat{x}_\text{stable},x^\star,\hat{x}_\text{base}^{(\ell)}),
\end{aligned}
\end{equation*}
where $$\kappa(\hat{x}_\text{stable},x^\star,\hat{x}_\text{base}^{(\ell)}) := \rho(\hat{x}^{(\ell)}_\text{base},x^\star)-\rk(\hat{x}_\text{base}^{(\ell)}) + \rho(\hat{x}_\text{stable},\hat{x}_\text{base}^{(\ell)})-\rho(\hat{x}_\text{stable},x^\star).$$

Since the choice of $l$ was arbitrary, we note that:
\begin{equation*}
    \begin{aligned}
        \mathrm{FD}(\hat{x}_\text{stable},x^\star) &= \frac{2}{B}\sum_{\ell=1}^{B/2}\min_{t \in \{0,1\}}\Bigg\{ \left[\rk(\hat{x}_\text{stable})-\rho(\hat{x}_\text{stable},\hatx_\text{base}^{(2\ell-t)})\right] + \left[\rk(\hatx_\text{base}^{(2\ell-t)})-\rho(\hatx_\text{base}^{(2\ell-t)},x^\star)\right] + \kappa(\hat{x}_\text{stable},x^\star,\hat{x}_\text{base}^{(2\ell-t)})\Bigg\}\\
        &\leq \frac{2}{B}\sum_{\ell=1}^{B/2}\min_{t \in \{0,1\}}\Bigg\{ \left[\rk(\hatx_\text{base}^{(2\ell-t)})-\rho(\hatx_\text{base}^{(2\ell-t)},x^\star)\right]\Bigg\}+\frac{2}{B}\sum_{\ell=1}^B \left[\rk(\hat{x}_\text{stable})-\rho(\hat{x}_\text{stable},\hatx_\text{base}^{(\ell)})\right]\\ &+ \frac{2}{B}\sum_{\ell=1}^B\kappa(\hat{x}_\text{stable},x^\star,\hat{x}_\text{base}^{(\ell)})\\
        &\leq \frac{2}{B}\sum_{\ell=1}^{B/2}\prod_{t\in\{0,1\}} \sqrt{\rk(\hatx_\text{base}^{(2\ell-t)})-\rho(\hatx_\text{base}^{(2\ell-t)},x^\star)}+ 2\alpha\rk(\hat{x}_\text{stable}) + \frac{2}{B}\sum_{\ell=1}^B\kappa(\hat{x}_\text{stable},x^\star,\hat{x}_\text{base}^{(\ell)}).
    \end{aligned}
\end{equation*}
Here, the second inequality follows from $\min\{a+b,c+d\} \leq \min\{a,c\}+b+d$ for $a,b,c,d\geq 0$. The third inequality follows from $\min\{a,b\} \leq \sqrt{ab}$ for $a,b \geq 0$ and 
\begin{equation}
    \begin{aligned}
        \frac{1}{B}\sum_{\ell=1}^B \rk(\hat{x}_\text{stable})-\rho(\hat{x}_\text{stable},\hat{x}_\text{base}^{(\ell)}) =  \sum_{k=1}^{\rk(\hat{x}_\text{stable})}\frac{1}{B}\sum_{\ell=1}^B 1-[\rho(x_k,\hatx_\text{base}^{(\ell)})-\rho(x_{k-1},\hatx_\text{base}^{(\ell)})] \leq \alpha\rk(\hat{x}_\text{stable}),
    \end{aligned}
    \label{eq:bound_rk_bags}
\end{equation}
where $(x_0,x_1,\dots,x_{\hat{k}})$ is a sequence specifying a path from the least element $x_0$ to $x_{\hat{k}} = \hat{x}_\text{stable}$ with  $\rk(\hat{x}_\text{stable}) = \hat{k}$. Thus, $\frac{1}{B}\sum_{\ell=1}^B \rho(\hat{x}_\text{stable},\hat{x}_\text{base}^{(\ell)}) \geq (1-\alpha)\rk(\hatx_\text{stable})$. As $\rho(\hat{x}_\text{stable},\hat{x}_\text{base}^{(\ell)}) \leq \rk(\hat{x}_\text{base}^{(\ell)})$, we can then conclude that $\mathbb{E}[\rk(\hat{x}_\text{stable})] \leq \frac{\mathbb{E}[\rk(\hat{x}_\text{sub})]}{1-\alpha}$. Taking expectations and using the fact that the data across complementary bags is IID, we obtain:
\begin{equation*}
\begin{aligned}
\mathrm{FD}(\hat{x}_\text{stable},x^\star) \leq \mathbb{E}[\sqrt{\FD(\hat{x}_\text{sub},x^\star)}]^2 + \frac{2\alpha}{1-\alpha}\mathbb{E}[\rk(\hat{x}_\text{sub})] + \frac{2}{B}\sum_{\ell=1}^B\mathbb{E}[\kappa(\hat{x}_\text{stable},x^\star,\hat{x}_\text{base}^{(\ell)})].
\end{aligned}
\end{equation*}
It remains to bound $\frac{2}{B}\sum_{\ell=1}^B\mathbb{E}[\kappa(\hat{x}_\text{stable},x^\star,\hat{x}_\text{base}^{(\ell)})]$ for subspace selection and blind-source separation.\\[0.2in]

\noindent\textbf{Subspace-selection}: We will use the similarity valuation $\rho := \rho_\text{\text{subspace}}$ in $\mathrm{Eq.}$ (3). Note that:
    \begin{equation}
        \begin{aligned}
            \rk(x)-\rho(x,y) = \tr\left(\Proj_x\Proj_{y^\perp}\right) &= \tr\left(\Proj_x\Proj_z\Proj_{y^\perp}\Proj_z\right)+\tr\left(\Proj_x\Proj_{z^\perp}\Proj_{y^\perp}\Proj_{z^\perp}\right)\\&+\tr\left(\Proj_x\Proj_{z^\perp}\Proj_{y^\perp}\Proj_z\right)+\tr\left(\Proj_x\Proj_z\Proj_{y^\perp}\Proj_{z^\perp}\right)\\
            &\leq \tr\left(\Proj_{y^\perp}\Proj_z\right)+\tr\left(\Proj_x\Proj_{z^\perp}\right) + \tr\left(\left[\Proj_x,\Proj_{z^\perp}\right]\left[\Proj_z,\Proj_{y^\perp}\right] \right) \\
            &= \rk(z)-\rho(y,z) + \rk(x)-\rho(x,z) + \tr\left(\left[\Proj_x,\Proj_{z^\perp}\right]\left[\Proj_z,\Proj_{y^\perp}\right] \right).
        \end{aligned}
        \label{eqn:bound_1}
    \end{equation}
Here, for matrices $A,B \in \mathbb{R}^{p \times p}$, $[A,B] = AB - BA$ represents the commutator.  
Furthermore, note that:
\begin{equation}
\begin{aligned}
\tr\left(\left[\Proj_x,\Proj_{z^\perp}\right]\left[\Proj_z,\Proj_{y^\perp}\right] \right)&\leq \|\left[\Proj_x,\Proj_{z^\perp}\right]\|_\star\|\left[\Proj_z,\Proj_{y^\perp}\right]\|_2\\
            &\leq 2\sqrt{\rk(x)}\sqrt{\rk(x)-\rho(x,z)} \|\left[\Proj_z,\Proj_{y}\right]\|_2.
\end{aligned}
\label{eqn:bound_2}
\end{equation}
Combining the bounds \eqref{eqn:bound_1} and \eqref{eqn:bound_2}, we find that: 
\begin{equation*}
\begin{aligned}
\rk(x)-\rho(x,y) &\leq \rk(z)-\rho(y,z) + \rk(x)-\rho(x,z) + 2\sqrt{\rk(x)}\sqrt{\rk(x)-\rho(x,z)} \|\left[\Proj_z,\Proj_{y}\right]\|_2 \\
&\leq  \rk(z)-\rho(y,z) + \rk(x)-\rho(x,z) + \sqrt{\rk(x)}\sqrt{\rk(x)-\rho(x,z)}.
\label{eqn:bound_3}
\end{aligned}
\end{equation*}
Here, the second inequality follows from the fact that for projection matrices $A$ and $B$, $\|[A,B]\|_2 \leq 1/2$. From this inequality, we conclude that in the subspace selection setting, 
\begin{equation*}
\begin{aligned}
\frac{1}{B}\sum_{\ell=1}^B \kappa(\hat{x}_\text{stable},x^\star,\hat{x}_\text{base}^{(\ell)}) &\leq \sqrt{\rk(\hat{x}_\text{stable})}\frac{1}{B}\sum_{l = 1}^{B}\sqrt{\rk(\hat{x}_\text{stable})-\rho(\hat{x}_\text{stable},\hat{x}_\text{base}^{(\ell)})} \\
&\leq \sqrt{\rk(\hat{x}_\text{stable})}\sqrt{\frac{1}{B}\sum_{l = 1}^{B}\rk(\hat{x}_\text{stable})-\rho(\hat{x}_\text{stable},\hat{x}_\text{base}^{(\ell)})} \\
&\leq \sqrt{\alpha}\rk(\hat{x}_\text{stable}).
\end{aligned}
\end{equation*}
Here, the second equality follows from Cauchy-Schwartz and the last inequality follows from the bound \eqref{eq:bound_rk_bags}. Recalling that $\mathbb{E}[\rk(\hat{x}_\text{stable})] \leq \frac{\mathbb{E}[\rk(\hat{x}_\text{sub})]}{1-\alpha}$, we obtain the final bound:
\begin{equation*}
\begin{aligned}
\mathrm{FD}(\hat{x}_\text{stable},x^\star) \leq \mathbb{E}[\sqrt{\FD(\hat{x}_\text{sub},x^\star)}]^2 + \frac{2\alpha+\sqrt{\alpha}}{1-\alpha}\mathbb{E}[\rk(\hat{x}_\text{sub})].
\end{aligned}
\end{equation*}

\noindent\textbf{Blind-source separation} We will use the similarity valuation $\rho := \rho_\text{\text{source-separation}}$ in $\mathrm{Eq.}$ (4). For simplicity of notation, associated with any element $z\in\Lp$, we consider a block-diagonal $p^2 \times p^2$ projection matrix where each $p\times{p}$ block is a projection matrix of the subspace spanned by a vector in $z$. We denote this projection matrix $\Proj_z$. Then, $\rho(x,y) = \max_{\Pi \in \mathbb{S}^{p^2}_{\text{block}}} \tr\left(\Proj_x\Pi\Proj_{y}\Pi^T\right)$ where $\mathbb{S}^{p^2}_{\text{block}}$ is the space of $p^2 \times p^2$ permutation matrices that are block-diagonal where each block is of size $p\times{p}$. \\

Note that:
\begin{equation*}
    \begin{aligned}
        \rk(x)-\rho(x,y) &= \min_{\Pi\in\mathbb{S}^{p^2}_\text{block}} \tr\left(\Proj_x\Pi\Proj_{y^\perp}\Pi^T\right)\\
&\leq\min_{\tilde{\Pi}\in\mathbb{S}^{p^2}_\text{block}}\min_{\Pi\in\mathbb{S}^{p^2}_\text{block}} \tr\left(\Pi\Proj_{y^\perp}\Pi^T\tilde{\Pi}\Proj_{z}\tilde{\Pi}^T\right)+\tr\left(\Proj_x\tilde{\Pi}\Proj_{z^\perp}\tilde{\Pi}^T\right)\\       
        &~~~~~~~~~~~~~~~~~~~~~+2\sqrt{\rk(x)}\sqrt{\tr\left(\Proj_x\tilde{\Pi}\Proj_{z^\perp}\tilde{\Pi}^T\right)}\|[\tilde{\Pi}\Proj_{z}\tilde{\Pi}^T,\Pi\Proj_{y}\Pi^T]\|_2 \\
        &\leq \min_{\tilde{\Pi}\in\mathbb{S}^{p^2}_\text{block}}\tr\left(\Proj_x\tilde{\Pi}\Proj_{z^\perp}\tilde{\Pi}^T\right)+2\sqrt{\rk(x)}\sqrt{\tr\left(\Proj_x\tilde{\Pi}\Proj_{z^\perp}\tilde{\Pi}^T\right)}\max_{\bar{\Pi},\bar{\tilde{\Pi}}\in\mathbb{S}^{p^2}_\text{block}}\|[\bar{\tilde{\Pi}}\Proj_{z}\bar{\tilde{\Pi}}^T,\bar{\Pi}\Proj_{y}\bar{\Pi}^T]\|_2 \\
&~~~~~+\max_{\tilde{\Pi}\in\mathbb{S}^p_\text{block}}\min_{\Pi\in\mathbb{S}^p_\text{block}}\tr\left(\Pi(\id-\Proj_{y})\Pi^T\tilde{\Pi}\Proj_{z}\tilde{\Pi}^T\right)\\ 
\\
        &= [\rk(x)-\rho(x,z] + [\rk(z)-\rho(z,y)]\\&+2\sqrt{\rk(x)}\sqrt{\rk(x)-\rho(x,z)} \max_{\bar{\Pi},\bar{\tilde{\Pi}}\in\mathbb{S}^{p^2}_\text{block}}\|[\bar{\tilde{\Pi}}\Proj_{z}\bar{\tilde{\Pi}}^T,\bar{\Pi}\Proj_{y}\bar{\Pi}^T]\|_2.
    \end{aligned}
\end{equation*}
Here, the first inequality follows from a similar analysis as arriving to \eqref{eqn:bound_1} in subspace selection. The second inequality follows from the fact that $\min_{a,b} f(a)+g(b) \leq \min_a f(a) + \max_b f(b)$. Note that projection matrices $A,B$, $[A,B] \leq \frac{1}{2}$. Then, following the same exact reasoning as the subspace case, we have that in the blind-source separation setting $\frac{1}{B}\sum_{\ell=1}^B \kappa(\hat{x}_\text{stable},x^\star,\hat{x}_\text{base}^{(\ell)}) \leq \sqrt{\alpha}\rk(\hat{x}_\text{sub})$. The result follows subsequently.

\section{Proof of Lemmas~\ref*{lemma:discrete_integer}-\ref*{lemma:two_covering_pairs}}
\label{sec:proof_of_main_paper_lemmas}
\begin{proof}[Proof of Lemma~\ref*{lemma:discrete_integer}]
Recall the telescoping sum decomposition $\mathrm{Eq}.$ (5) that $\FD(x_k,x^\star) = \sum_{i=1}^k 1-[f(x_{i-1},x_i;x^\star)]$. From the first property of similarity valuation that it yields non-negative values, second property of similarity valuation that $\rho(x,y)\leq \rho(z,y)$ for $x \preceq z$, and that the $\rho$ is an integer-valued similarity valuation, we have that $\FD(x,x^\star) \leq \sum_{i=1}^k \mathbb{I}[(x_{i-1},x_i) \in \mathcal{T}_{\text{null}}]$.
\end{proof}
\begin{proof}[Proof of Lemma~\ref*{lemma:two_covering_pairs}]
 For any covering pairs $(x,y)$ and $(u,v)$ with $v \preceq x$, we cannot have that $f(x,y; z) = f(u,v; z)$ for all $z \in \Lp$. Suppose as a point of contradiction that for every $z \in \Lp$, $f(x,y;z)=f(u,v;z)$. Let $z = v$. Then, by the third property of a similarity valuation (see Definition~\ref*{defn:rho_prop}), $\rho(u,z) = \rk(u)$ and $\rho(v,z)= \rk(v)$; thus, for this choice of $z$, $f(u,v;z)=1$. On the other hand, again by the third property of a similarity valuation and for the choice of $z=v$, since $u \preceq v \preceq x \preceq y$,  $\rho(x,z)=\rho(y,z)=\rk(v)$ and thus $f(x,y;z) = 0$.  
 \end{proof}

\end{document}